\documentclass[journal,twoside,web]{ieeecolor}   

\let\labelindent\relax 
\usepackage{amsmath,amsfonts,amssymb,amsthm}
\usepackage{mathrsfs}
\usepackage{mathtools}
\usepackage{array}
\usepackage{float}
\usepackage{extarrows}
\usepackage{cite}
\usepackage{url}
\usepackage{color}
\usepackage[table]{xcolor}
\usepackage[colorlinks,linkcolor=orange,citecolor=blue]{hyperref}
\usepackage{algorithmic}
\usepackage{graphicx}
\usepackage{generic}   
\usepackage{textcomp}
\usepackage[font=footnotesize]{caption}
\usepackage[font=footnotesize]{subcaption}
\usepackage{stfloats}
\usepackage{epstopdf}
\usepackage{arydshln}
\usepackage{ifthen}
\usepackage{pict2e}
\usepackage[T1]{fontenc}
\usepackage{enumitem}
\usepackage{tikz}
\usepackage{diagbox}
\usepackage{multirow}

\allowdisplaybreaks[1] 

\DeclareGraphicsExtensions{.pdf,.png}
\DeclareMathOperator{\blkdiag}{blkdiag}

\DeclareMathOperator{\trace}{trace}

\DeclareMathOperator{\vect}{vec}
\DeclareMathOperator{\range}{range}
\DeclareMathOperator{\spec}{spec}
\DeclareMathOperator{\image}{image}

\definecolor{myblue}{RGB}{231, 245, 254}

\newtheorem{theorem}{Theorem}

\renewcommand{\t}{^{\mbox{\tiny\sf T}}} 

\newcommand{\R}{\mathbb{R}}

\def\send#1#2{\stackrel{#1}{\hbox to #2{\rightarrowfill}}}
\def\-{\!\!\!\!\!-}

\newcommand{\rank}{{\rm rank\;}}

\newtheorem{lemma}{Lemma}
\newtheorem{remark}{Remark}
\newtheorem{proposition}{Proposition}
\newtheorem{corollary}{Corollary}
\newtheorem{property}{Property}
\newtheorem{assumption}{Assumption}

\def\R{{\rm I\!R}} 

\newcounter{seqn}[equation]
\def\theseqn{\arabic{equation}\alph{seqn}}

\def\endseqn{\eqno \@seqnnum
$$\ignorespaces}
\def\@seqnnum{(\theseqn)}

\newskip\mcentering \mcentering=0pt plus 1000pt minus 1000pt

\def\meqalignno#1{
\halign to\displaywidth{
    \hbox to 0pt{\kern\displaywidth\llap{$##$}\hss}\tabskip=\mcentering
    &\hfil$\displaystyle{##}$\tabskip=\mcentering
   &&$\displaystyle{{}##}$\hfil\tabskip=\mcentering
    \crcr
    #1\crcr}}

\def\dspace{\multiply\normalbaselineskip 150
		  \divide\normalbaselineskip 100 \normalbaselines
		  \csname @@normalbaselineskip\endcsname\normalbaselineskip}
\def\sspace{\multiply\normalbaselineskip 200
		 \divide\normalbaselineskip 300 \normalbaselines
		 \csname @@normalbaselineskip\endcsname\normalbaselineskip}
\def\sdspace{\multiply\normalbaselineskip 160
		 \divide\normalbaselineskip 150 \normalbaselines
		 \csname @@normalbaselineskip\endcsname\normalbaselineskip}

\def\@{\tilde}

\def\3dot#1{\buildrel\textstyle...\over#1}

\renewcommand{\R}{\mathbb{R}}

\def\BibTeX{{\rm B\kern-.05em{\sc i\kern-.025em b}\kern-.08em
    T\kern-.1667em\lower.7ex\hbox{E}\kern-.125emX}}

\begin{document}

\title{Optimal Covariance Steering for Discrete-Time Linear Stochastic Systems}

\author{Fengjiao Liu, George Rapakoulias, and Panagiotis Tsiotras
\thanks{F. Liu is with the Department of Electrical and Computer Engineering, FAMU-FSU College of Engineering, Tallahassee, FL 32310 USA (e-mail: fliu@eng.famu.fsu.edu). 
G. Rapakoulias and P. Tsiotras are with the School of Aerospace Engineering, Georgia Institute of Technology, Atlanta, GA 30332 USA (e-mail: \{grap, tsiotras\}@gatech.edu).}}

\maketitle
\thispagestyle{empty} 
\pagestyle{empty}     

\begin{abstract}
In this paper, we study the optimal control problem for steering the state covariance of a discrete-time linear stochastic system over a finite time horizon. 
First, we establish the existence and uniqueness of the optimal control law for a quadratic cost function. 
Then, we show the separation of the optimal mean and the covariance steering problems. 
We also develop efficient computational methods to solve for the optimal control law, which is identified as the solution to a semi-definite program. 
The effectiveness of the proposed approach is demonstrated through numerical examples. 
In the process, we also obtain some novel theoretical results for a matrix Riccati difference equation, which may be of independent interest. 
\end{abstract}

\begin{IEEEkeywords}
Covariance steering, semi-definite program, Riccati difference equation.
\end{IEEEkeywords}

\section{Introduction}

In recent years, the need to quantify and control the uncertainty in physical systems has prompted a burgeoning interest in studying the evolution of the distribution of the trajectories of stochastic systems. 
A special case of this point of view is covariance control, the earliest research of which can be traced back to a series of articles from 1985 onward on the assignability of the state covariance via state feedback over an infinite time horizon~\cite{hotz1985covariance, collins1985covariance, collins1987theory, hsieh1990all, xu1992improved, zhang2016parametric}. 
The optimal control that assigns a prescribed stationary state covariance with minimum control energy was developed in~\cite{grigoriadis1997minimum}. 
More recent studies address the problem of optimally steering the state covariance of a linear stochastic system over a finite time horizon~\cite{bakolas2018finite, okamoto2019input, kotsalis2021convex, pilipovsky2021covariance, balci2023covariance}, including probabilistic (chance) constraints.

For discrete-time optimal covariance steering, the approach taken by most of the current work involves three steps: 
1) reformulate the problem as an optimization problem in an augmented state space that includes the entire history of the state; 
2) relax the non-convex chance constraints to convex constraints for a tractable convex optimization problem; 
3) solve the convex optimization problem numerically to approximate the optimal control. 
One of the main reasons for the above approach is that, when chance constraints on the sample paths of the input and the state are present, the state-mean and state-covariance constraints are coupled, which makes it difficult to find the optimal control analytically~\cite{bakolas2018finite, okamoto2019input, pilipovsky2021covariance}. 
When there are no chance constraints, the desired terminal covariance can be replaced with a soft constraint on the Wasserstein distance between the desired and the actual terminal Gaussian distributions, which can be solved using a randomized state feedback control in terms of a (convex) semi-definite program (SDP)~\cite{balci2022exact}. 
Finite-horizon covariance control has also been applied to a model-predictive-control setting~\cite{okamoto2019stochastic, yin2022trajectory}, in which, at each time step, an optimal covariance steering problem is solved in a receding-horizon fashion.

In addition to controlling the first two moments of the state of a stochastic system, it is also possible to steer the entire state distribution using optimization techniques~\cite{sivaramakrishnan2022distribution}. 
The continuous-time counterpart for the finite-horizon covariance steering problem is investigated in~\cite{chen2016I, chen2016II, chen2018III, liu2024add, liu2024mult}. 
By quantifying the uncertainty directly, covariance control theory can be applied to various practical scenarios, such as spacecraft landing~\cite{ridderhof2018uncertainty}, spacecraft trajectory optimization~\cite{ridderhof2020fuel}, vehicle path planning~\cite{okamoto2019optimal, yin2022trajectory}, and aircraft motion planning~\cite{zheng2022belief}.

Despite the success in finding an approximate solution to the optimal covariance steering problem using convex optimization, little is known regarding fundamental questions, such as the existence and uniqueness of the optimal control. 
Additionally, very little is known regarding the conservativeness of the approximate solution compared to the optimal control, if one exists. 
In this paper, we provide answers to these questions for the discrete-time covariance steering case.

The contributions and technical challenges of this work can be summarized as follows:
We first establish an analytical result (Theorem~\ref{thm:exist-unique-Q}) on the existence and uniqueness of the optimal control law for steering the state covariance of a discrete-time linear stochastic system with respect to a quadratic cost function. 
In the process of proving Theorem~\ref{thm:exist-unique-Q}, the main technical challenge is to derive useful necessary and sufficient conditions (Theorem~\ref{thm:pi-pos-exist-cond}) for a critical property (Property~\ref{pty:pi-pos}) of a matrix Riccati difference equation. 
Then, in the absence of any chance constraints, we demonstrate that the optimal state mean and covariance can be controlled independently (Theorem~\ref{thm:mean-covar}). 
Lastly, we show that the optimal control law can be computed by reformulating
the problem as an SDP (Theorem~\ref{thm:prob_equiv}), and that the control
law takes a state-feedback form (Theorem~\ref{thm:lossless}). 
%
To the best of our knowledge, this is the first work to analyze the existence and uniqueness of the optimal control for steering the state covariance to any given terminal value in discrete time and to show that the optimal control can be solved exactly.

The rest of the paper is organized as follows. 
The discrete-time covariance steering problem considered in this paper is formulated in Section~\ref{sec:prob}. 
The existence and uniqueness of the optimal control as well as the separation of the optimal mean and covariance steering problems are shown in Section~\ref{sec:exist-unique}. 
The methods for computing the optimal control are provided in Section~\ref{sec:compute}. 
A numerical example is presented in Section~\ref{sec:example}. 
For conciseness and ease of exposition, most of the proofs and auxiliary results are given in the Appendices.

\section{Problem Formulation} \label{sec:prob}

Consider the linear time-varying stochastic system corrupted by noise,
\begin{equation} \label{sys:disc-time}
x_{k+1} = A_{k} x_{k} + B_{k} u_{k} + D_{k} w_{k}, \quad k = 0, 1, \ldots, N-1,
\end{equation}
where $x_{k} \in \mathbb{R}^{n}$ is the state, $u_{k} \in \mathbb{R}^{p}$ is the control input at time step $k$, 
$w_{k} \in \mathbb{R}^{q}$ is the square-integrable noise independent of $x_{k}$ and $u_{k}$, such that $\mathbb{E}\left[w_{k}\right] = 0$ and $\mathbb{E}\left[w_{k} w_{k}\t\right] = I_{q}$, 
and $A_{k} \in \mathbb{R}^{n \times n}$, $B_{k} \in \mathbb{R}^{n \times p}$, and $D_{k} \in \mathbb{R}^{n \times q}$ are the system coefficient matrices.

The initial state $x_{0}$ and the desired terminal state $x_{N}$ are characterized by their mean and their covariance matrices given by 
\begin{subequations} \label{cst:mean-cov-Q}
\begin{align}
&\mathbb{E}\left[x_{0}\right] = \bar{\mu}_{0}, 
\quad 
\mathbb{E}\big[(x_{0} - \bar{\mu}_{0}) (x_{0} - \bar{\mu}_{0})\t \big] = \bar{\Sigma}_{0} \succ 0, \label{bgn:mean-cov-Q} 
\\
&\mathbb{E}\left[x_{N}\right] = \bar{\mu}_{N}, 
\nonumber \\
&\mathbb{E}\big[(x_{N} - \bar{\mu}_{N}) (x_{N} - \bar{\mu}_{N})\t \big] = \bar{\Sigma}_{N} \succ D_{N-1} D_{N-1}\t. \label{tgt:mean-cov-Q}
\end{align}
\end{subequations}

\begin{assumption} \label{asm:A-inv}
For all $k = 0, 1, \ldots, N-1$, the matrix $A_{k}$ is invertible.\footnote{This assumption requires that the system $x_{k+1} = A_{k} x_{k} + B_{k} u_{k}$ is time reversible. 
This is a reasonable assumption, as it follows from discretizing the system $\dot{x}(t) = A(t) x(t) + B(t) u(t)$ that $A_{k} = I_{n} + A(t) \delta t$, where $\delta t$ is the discrete time step, and
for a sufficiently small $\delta t$, $A_{k}$ is invertible.}
\end{assumption}

With Assumption~\ref{asm:A-inv}, we can define the state transition matrix $\Phi_{A}(k, \ell)$ from time $\ell$ to $k$ as
\begin{align*}
\Phi_{A}(k, \ell) \triangleq 
\begin{cases}
A_{k-1} A_{k-2} \cdots A_{\ell}, & 0 \leq \ell < k \leq N, 
\\* 
I_{n}, & 0 \leq \ell = k \leq N, 
\\* 
A_{k}^{-1} A_{k+1}^{-1} \cdots A_{\ell-1}^{-1}, & 0 \leq k < \ell \leq N.
\end{cases}
\end{align*}
The reachability Gramian $G(k, \ell)$ of \eqref{sys:disc-time} from time $\ell$ to $k$ is defined as
\begin{align*}
&G(k,\ell) \triangleq 
\\*
&\begin{cases}
\sum_{i=\ell}^{k-1} \Phi_{A}(k, i+1) B_{i} B_{i}\t \Phi_{A}(k, i+1)\t, &\hspace{-1.5mm} 0 \leq \ell < k \leq N, 
\\* 
0_{n \times n}, &\hspace{-1.5mm} 0 \leq \ell = k \leq N, 
\\* 
- \sum_{i=k}^{\ell-1} \Phi_{A}(k, i+1) B_{i} B_{i}\t \Phi_{A}(k, i+1)\t, &\hspace{-1.5mm} 0 \leq k < \ell \leq N.
\end{cases}
\end{align*}

\begin{assumption} \label{asm:G-inv}
System~\eqref{sys:disc-time} is controllable from time $0$ to $N$. That is, the reachability Gramian $G(N, 0) \succ 0$.
\end{assumption}

\begin{assumption} \label{asm:G-D}
For all $k = 1, 2, \ldots, N-1$, $\range \; \Phi_{A}(N, k) D_{k-1} \subseteq \range \; G(N, k)$.
\end{assumption}



We point out that, for system~\eqref{sys:disc-time}, Assumption~\ref{asm:G-D} together with Assumption~\ref{asm:G-inv} is a necessary and sufficient condition for all terminal state covariances $\bar{\Sigma}_{N} \succ D_{N-1} D_{N-1}\t$ to be reachable from a given initial state covariance $\bar{\Sigma}_{0} \succ 0$~\cite[Theorem~2]{liu2024reach}.

The control input $\{u_{k}\}_{k=0}^{N-1}$, is said to be \emph{admissible} if each $u_{k}$ depends only on $k$ and, perhaps, on the past history of the states $x_{0}, x_{1}, \ldots, x_{k}$, 
such that the desired boundary constraints \eqref{cst:mean-cov-Q} are satisfied. 
Without loss of generality, we may assume that the desired mean terminal state is $\bar{\mu}_{N} = 0$. 
In this case, we introduce a cost functional defined by
\begin{equation} \label{cost-Q}
J \triangleq \mathbb{E}\left[\sum_{k=0}^{N-1} x_{k}\t Q_{k} x_{k} + u_{k}\t R_{k} u_{k} \right],
\end{equation}
where $Q_{k} \succeq 0$ and $R_{k} \succ 0$ are matrices of dimensions $n \times n$ and $p \times p$, respectively.

\textbf{Optimal Covariance Steering Problem:}
Determine the optimal admissible control $u_{k}^{*}$ that minimizes the quadratic cost \eqref{cost-Q} subject to the initial and terminal state constraints \eqref{cst:mean-cov-Q}.

\section{Existence \& Uniqueness of the Optimal Control} \label{sec:exist-unique}

The main theorem on the existence and uniqueness of the optimal control is summarized below.

\begin{theorem} \label{thm:exist-unique-Q}
Let $\bar{\Sigma}_{0} \succ 0$ and $\bar{\Sigma}_{N} \succ D_{N-1} D_{N-1}\t$. 
Under Assumptions~\ref{asm:A-inv}-\ref{asm:G-D}, the unique optimal control law that solves the covariance steering problem for system~\eqref{sys:disc-time} is given, for all $k = 0, 1, \ldots, N-1$, by 
\begin{align} \label{ctrl:opt-Q}
u^{*}_{k} &= 
- \big(R_{k} + B_{k}\t \Pi_{k+1} B_{k}\big)^{-1} B_{k}\t \Pi_{k+1} A_{k} x_{k} 
\nonumber \\*
&\hspace{5mm} - \big(R_{k} + B_{k}\t \Pi_{k+1} B_{k}\big)^{-1} B_{k}\t \Phi_{\bar{A}}(N, k+1)\t 
\nonumber \\*
&\hspace{5mm} \times \Bigg( \sum_{i=0}^{N-1} \Phi_{\bar{A}}(N, i+1) B_{i} \big(R_{i} + B_{i}\t \Pi_{i+1} B_{i} \big)^{-1} 
\nonumber \\*
&\hspace{5mm} \times B_{i}\t \Phi_{\bar{A}}(N, i+1)\t \Bigg)^{-1} 
\Big( \Phi_{\bar{A}}(N, 0) \bar{\mu}_{0} - \bar{\mu}_{N} \Big), 
\end{align}
where $\Pi_{k}$ is the unique solution of the coupled matrix difference equations 
\begin{subequations} \label{eqn:coupled}
\begin{align}
&\Pi_{k} = A_{k}\t \Pi_{k+1} \big(I_{n} + B_{k} R_{k}^{-1} B_{k}\t \Pi_{k+1}\big)^{-1} A_{k} + Q_{k}, \label{eqn:pi-Q} 
\\
&\Sigma_{k+1} = \big(I_{n} + B_{k} R_{k}^{-1} B_{k}\t \Pi_{k+1}\big)^{-1} A_{k} \Sigma_{k} A_{k}\t 
\nonumber \\*
&\hspace{15mm} \times \big(I_{n} + \Pi_{k+1} B_{k} R_{k}^{-1} B_{k}\t \big)^{-1} + D_{k} D_{k}\t, \label{eqn:sigma-Q} 
\\
&\Sigma_{0} = \bar{\Sigma}_{0}, \quad \Sigma_{N} = \bar{\Sigma}_{N}. \label{bdr:sigma-Q}
\end{align}
\end{subequations}
Furthermore,
for all $k = 0, 1, \ldots, N-1$, it holds that $R_{k} + B_{k}\t \Pi_{k+1} B_{k} \succ 0$,  
where $\Phi_{\bar{A}}(\ell, s)$ is the state transition matrix of 
\begin{equation} \label{def:A-bar}
\bar{A}_{k} \triangleq \big(I_{n} + B_{k} R_{k}^{-1} B_{k}\t \Pi_{k+1}\big)^{-1} A_{k}, \quad k = 0, 1, \ldots, N-1,
\end{equation}
from time $s$ to $\ell$.
\end{theorem}


\begin{remark}
It will be shown in Section~\ref{subsec:separation} that the optimal control \eqref{ctrl:opt-Q} can be written in the form $u_{k}^{*} = K_{k} \big(x_{k} - \mu_{k}\big) + v_{k}$, where $K_{k}$ is the feedback gain matrix, $\mu_{k}$ is the state mean, and $v_{k}$ is the feed-forward term. 
\end{remark}

\begin{remark}
Notice that, unlike the case of the Kalman filter or the linear-quadratic regulator problems, the Riccati equation \eqref{eqn:pi-Q} does not have an initial or terminal condition, while the Lyapunov equation \eqref{eqn:sigma-Q} has split boundary conditions \eqref{bdr:sigma-Q}. 
This distinction poses major challenges in solving the coupled equations \eqref{eqn:coupled}. 
\end{remark}

For ease of reference, an important property of the Riccati equation \eqref{eqn:pi-Q} in Theorem~\ref{thm:exist-unique-Q} is summarized below.

\begin{property} \label{pty:pi-pos}
For all $k = 0, 1, \ldots, N-1$,
the matrix $R_{k} + B_{k}\t \Pi_{k+1} B_{k} \succ 0$. 
\end{property}

At this point, it suffices to note that Property~\ref{pty:pi-pos} does not follow immediately from $R_{k} \succ 0$, since $\Pi_{k+1}$ is not necessarily positive semi-definite. 
Necessary and sufficient conditions for Property~\ref{pty:pi-pos} to hold are provided by Theorem~\ref{thm:pi-pos-exist-cond} in Section~\ref{subsec:riccati}.

\begin{remark}
Since $R_{k} \succ 0$, Property~\ref{pty:pi-pos} implies that,
for all $k = 0, 1, \ldots, N-1$, 
the matrix $I_{n} + B_{k} R_{k}^{-1} B_{k}\t \Pi_{k+1}$ is invertible according to the Sherman–Morrison–Woodbury formula~\cite{horn2012matrix}. 
\end{remark}


\begin{remark}
In the continuous-time case, when $D(t) D(t)\t \equiv \kappa B(t) R^{-1}(t) B(t)\t$ for some $\kappa > 0$, 
there exists a closed-form solution of $\Pi(0)$ for the continuous-time counterpart of the coupled matrix equations \eqref{eqn:coupled}~\cite{chen2016I, chen2018III}. 
In the discrete-time case, however, there may not be a closed-form solution for $\Pi_{0}$ even when $D_{k} D_{k}\t \equiv B_{k} R_{k}^{-1} B_{k}\t$. 
Interestingly, a closed-form solution exists for the discrete-time maximum entropy optimal covariance control problem~\cite{ito2023maximum}. 
\end{remark}


The rest of this section is dedicated to proving Theorem~\ref{thm:exist-unique-Q} and the separation of the mean and covariance steering problems, for which several intermediate results are needed. 
First, we show that \eqref{ctrl:opt-Q} is a candidate optimal control law, provided the coupled matrix difference equations \eqref{eqn:coupled} admit a solution that satisfies Property~\ref{pty:pi-pos}. 
Then, we give necessary and sufficient conditions for the Riccati difference equation \eqref{eqn:pi-Q} to have a unique solution, which satisfies Property~\ref{pty:pi-pos}. 
Next, we show the existence and uniqueness of the solution to the coupled equations \eqref{eqn:coupled}, which completes the proof of Theorem~\ref{thm:exist-unique-Q}. 
Finally, we show the separation of mean and covariance steering problems. 

Most of the proofs in this section are provided in the Appendices. 
Specifically, Appendix~A gives two auxiliary results which are eventually used in the proofs of all the main results of this section. 
Appendix~B shows some properties of the Riccati difference equation. 
Appendix~C establishes some useful equalities and inequalities of the state transition matrix of the coefficient matrices of the Riccati difference equation. 
Appendix~D proves several equivalent conditions for the Riccati difference equation to satisfy Property~\ref{pty:pi-pos}. 
Appendix~E presents two auxiliary results for showing the main results of this section.

\subsection{Candidate Optimal Control} \label{subsec:opt-ctrl}

A candidate optimal control law is derived using a ``completion of squares'' argument. 
To this end,
let $\Pi_{0}$, $\Pi_{1}$, $\ldots$, $\Pi_{N} \in \mathbb{R}^{n \times n}$ be symmetric matrices, and let $\psi_{0}$, $\psi_{1}$, $\ldots$, $\psi_{N} \in \mathbb{R}^{n}$ be adjoint state vectors to be defined below. 
In view of \eqref{cst:mean-cov-Q}, the expected values $\mathbb{E}\big[x_{0}\t \Pi_{0} x_{0}\big]$, $\mathbb{E}\big[x_{N}\t \Pi_{N} x_{N}\big]$, $\mathbb{E}\big[\psi_{0}\t x_{0}\big]$, and $\mathbb{E}\big[\psi_{N}\t x_{N}\big]$ are independent of the control. 
Since adding to the original cost function \eqref{cost-Q} extra terms that are independent of the control does not change the optimal control, we obtain an equivalent optimization problem with cost function
\begin{align*} 
&\tilde{J} = \mathbb{E}\left[\sum_{k=0}^{N-1} x_{k}\t Q_{k} x_{k} + u_{k}\t R_{k} u_{k}\right] 
\\*
&\hspace{5mm} + \mathbb{E}\big[x_{N}\t \Pi_{N} x_{N} - x_{0}\t \Pi_{0} x_{0} + 2\psi_{N}\t x_{N} - 2\psi_{0}\t x_{0}\big] 
\\
&= \mathbb{E}\Bigg[\sum_{k=0}^{N-1} \big(u_{k}\t + x_{k}\t A_{k}\t \Pi_{k+1} B_{k} S_{k}^{-1} + \psi_{k+1}\t B_{k} S_{k}^{-1}\big) S_{k} 
\\*
&\hspace{18mm} \times \big(u_{k} + S_{k}^{-1} B_{k}\t \Pi_{k+1} A_{k} x_{k} + S_{k}^{-1} B_{k}\t \psi_{k+1}\big) \Bigg] 
\\*
&\hspace{5mm} 
+ \sum_{k=0}^{N-1} \Big[\trace \big(\Pi_{k+1} D_{k} D_{k}\t \big) - \psi_{k+1}\t B_{k} S_{k}^{-1} B_{k}\t \psi_{k+1}\Big], 
\end{align*}
where 
$S_{k} \triangleq R_{k} + B_{k}\t \Pi_{k+1} B_{k}$, 
$\Pi_{k}$ satisfies the difference equation \eqref{eqn:pi-Q}, 
and $\psi_{k}$ satisfies the difference equation 
\begin{equation} \label{eqn:psi-Q}
\psi_{k} = A_{k}\t \big(I_{n} + \Pi_{k+1} B_{k} R_{k}^{-1} B_{k}\t \big)^{-1} \psi_{k+1}. 
\end{equation}
Since $\sum_{k=0}^{N-1} \big[\trace \big(\Pi_{k+1} D_{k} D_{k}\t \big) - \psi_{k+1}\t B_{k} S_{k}^{-1} B_{k}\t \psi_{k+1}\big]$ is independent of the control, a candidate optimal control is given, for all $k = 0, 1, \ldots, N-1$, as follows
\begin{multline} \label{ctrl:opt-Q-psi}
u^{*}_{k} = 
- \big(R_{k} + B_{k}\t \Pi_{k+1} B_{k}\big)^{-1} B_{k}\t \Pi_{k+1} A_{k} x_{k} 
\\* 
- \big(R_{k} + B_{k}\t \Pi_{k+1} B_{k}\big)^{-1} B_{k}\t \psi_{k+1}, 
\end{multline}
provided that $S_{k} \succ 0$, that is, if Property~\ref{pty:pi-pos} holds. 
Hence, for all $k = 0, 1, \ldots, N-1$, the corresponding optimal process is given by 
\begin{align} \label{sys:disc-time-opt-psi}
x_{k+1}^{*} &= \big(I_{n} + B_{k} R_{k}^{-1} B_{k}\t \Pi_{k+1}\big)^{-1} A_{k} x_{k}^{*} 
\nonumber \\* 
&\hspace{4mm} - B_{k} \big(R_{k} + B_{k}\t \Pi_{k+1} B_{k} \big)^{-1} B_{k}\t \psi_{k+1} + D_{k} w_{k}.
\end{align}
Let $\mu_{k} \triangleq \mathbb{E}[x_{k}^{*}]$ be the state mean at time $k = 0, 1, \ldots, N$, which satisfies the difference equation
\begin{multline} \label{eqn:mu-Q}
\mu_{k+1} = \big(I_{n} + B_{k} R_{k}^{-1} B_{k}\t \Pi_{k+1}\big)^{-1} A_{k} \mu_{k} 
\\* 
- B_{k} \big(R_{k} + B_{k}\t \Pi_{k+1} B_{k} \big)^{-1} B_{k}\t \psi_{k+1},
\end{multline}
with boundary conditions $\mu_{0} = \bar{\mu}_{0}$ and $\mu_{N} = \bar{\mu}_{N}$. 
In light of \eqref{eqn:psi-Q} and \eqref{def:A-bar}, we can rewrite equation \eqref{eqn:mu-Q} as 
\begin{equation*} 
\mu_{k+1} = \bar{A}_{k} \mu_{k} - B_{k} \big(R_{k} + B_{k}\t \Pi_{k+1} B_{k} \big)^{-1} B_{k}\t \Phi_{\bar{A}}(N, k+1)\t \psi_{N}. 
\end{equation*}
From the above equation and the boundary conditions $\mu_{0} = \bar{\mu}_{0}$ and $\mu_{N} = \bar{\mu}_{N}$, we obtain
\begin{align} \label{eqn:mu-N-0}
\bar{\mu}_{N} \! &= \! \Phi_{\bar{A}}(N, 0) \bar{\mu}_{0} \! - \hspace{-1.2mm} \sum_{i=0}^{N-1} \Phi_{\bar{A}}(N, i+1) B_{i} \big(R_{i} + B_{i}\t \Pi_{i+1} B_{i} \big)^{-1} 
\nonumber \\* 
&\hspace{35mm} \times B_{i}\t \Phi_{\bar{A}}(N, i+1)\t \psi_{N}. 
\end{align}

To proceed, we make use of the fact that 
\begin{equation} \label{eq:Lemma1}
\sum_{i=0}^{N-1} \Phi_{\bar{A}}(N, i+1) B_{i} \big(R_{i} + B_{i}\t \Pi_{i+1} B_{i} \big)^{-1} B_{i}\t \Phi_{\bar{A}}(N, i+1)\t \succ 0, 
\end{equation}
which is proved in Lemma~\ref{lem:A-bar-B-contr} in Appendix~A.
From equations \eqref{eqn:mu-N-0} and \eqref{eq:Lemma1} we obtain an expression of $\psi_{N}$. 
It then follows from \eqref{eqn:psi-Q} that 
\begin{align} \label{sol:psi-Q}
&\psi_{k+1} = 
\nonumber \\* 
&\Phi_{\bar{A}}(N, k+1)\t 
\! \Bigg( \sum_{i=0}^{N-1} \Phi_{\bar{A}}(N, i+1) B_{i} \big(R_{i} + B_{i}\t \Pi_{i+1} B_{i} \big)^{-1} B_{i}\t 
\nonumber \\* 
&\hspace{15mm} 
\times \Phi_{\bar{A}}(N, i+1)\t \! \Bigg)^{-1} 
\hspace{-1mm} \Big( \Phi_{\bar{A}}(N, 0) \bar{\mu}_{0} - \bar{\mu}_{N} \Big). 
\end{align}
Plugging \eqref{sol:psi-Q} into \eqref{ctrl:opt-Q-psi} and \eqref{sys:disc-time-opt-psi} yields \eqref{ctrl:opt-Q}.

It is clear from \eqref{sys:disc-time-opt-psi} and \eqref{eqn:mu-Q} that 
\begin{equation*}
x_{k+1}^{*} - \mu_{k+1} = \big(I_{n} + B_{k} R_{k}^{-1} B_{k}\t \Pi_{k+1}\big)^{-1} A_{k} (x_{k}^{*} - \mu_{k}) + D_{k} w_{k}. 
\end{equation*}
Therefore, the state covariance $\Sigma_{k} \triangleq \mathbb{E}\big[(x_{k}^{*} - \mu_{k}) (x_{k}^{*} - \mu_{k})\t \big]$ at time $k = 0, 1, \ldots, N$ satisfies the difference equation \eqref{eqn:sigma-Q} and the boundary conditions \eqref{bdr:sigma-Q} in light of \eqref{cst:mean-cov-Q}.

We have thus shown the following result.

\begin{proposition} \label{prp:suff-Q}
Under Assumptions~\ref{asm:A-inv} and \ref{asm:G-inv}, 
let $\Pi_{k}, \Sigma_{k} \in \mathbb{R}^{n \times n}$ be such that, for $k = 0, 1, \ldots, N-1$, Property~\ref{pty:pi-pos} and the coupled equations \eqref{eqn:coupled} hold. 
Then, the state feedback control $u^{*}$ given by \eqref{ctrl:opt-Q} is optimal for the covariance steering problem. 
Furthermore, if there exists a unique solution to the coupled equations \eqref{eqn:coupled}, then, the optimal control $u^{*}$ is unique. 
\end{proposition}


Since $R_{k} \succ 0$ can always be absorbed into $B_{k}$ by re-defining $\hat{B}_{k} \triangleq B_{k} R_{k}^{-\frac{1}{2}}$, without loss of generality and for the sake of notational simplicity for the rest of the paper, unless stated otherwise, we assume $R_{k} \equiv I_{p}$.

\subsection{Riccati Difference Equation} \label{subsec:riccati}

In this section, we give a necessary and sufficient condition for the unique solution $\Pi_{k}$ of the Riccati difference equation \eqref{eqn:pi-Q} to exist from time $0$ to $N$. 
A stronger condition for the solution of \eqref{eqn:pi-Q} to also satisfy Property~\ref{pty:pi-pos} is provided as well. 
To this end, let, 
\begin{equation} \label{def:M}
M_{k} \hspace{-0.6mm} \triangleq \hspace{-0.6mm}
\begin{bmatrix}
A_{k} + B_{k} B_{k}\t A_{k}^{- \mbox{\tiny\sf T}} Q_{k} & -B_{k} B_{k}\t A_{k}^{- \mbox{\tiny\sf T}} \\
-A_{k}^{- \mbox{\tiny\sf T}} Q_{k} & A_{k}^{- \mbox{\tiny\sf T}}
\end{bmatrix} 
\hspace{-0.8mm} \triangleq \hspace{-0.8mm}
\begin{bmatrix}
M_{k, 1} & M_{k, 2} \\
M_{k, 3} & M_{k, 4}
\end{bmatrix}
\end{equation}
It can be checked that $M_{k}$ is invertible, since $A_{k}^{- \mbox{\tiny\sf T}}$ and its Schur complement $A_{k} + B_{k} B_{k}\t A_{k}^{- \mbox{\tiny\sf T}} Q_{k} - B_{k} B_{k}\t A_{k}^{- \mbox{\tiny\sf T}} A_{k}\t A_{k}^{- \mbox{\tiny\sf T}} Q_{k} = A_{k}$ in $M_{k}$ are invertible. 
Let $\Phi_{M}(k, s)$ denote the state transition matrix of $M_{k}$ from time $s$ to $k$ decomposed as follows
\begin{align} \label{eqn:phi-M-blk}
\Phi_{M}(k, s) &\triangleq 
\begin{bmatrix}
\Phi_{11}(k, s) & \Phi_{12}(k, s) \\
\Phi_{21}(k, s) & \Phi_{22}(k, s)
\end{bmatrix}. 
\end{align}

First, we give the condition for the existence of a solution to the matrix equation \eqref{eqn:pi-Q}.

\begin{proposition} \label{prp:pi-Q-exist}
Suppose Assumption~\ref{asm:A-inv} holds. 
Given an initial condition $\Pi_{0}$, \eqref{eqn:pi-Q} has a solution from time $0$ to $N$ if and only if, for all $k = 0, 1, \ldots, N$, the matrix $\Phi_{11}(k, 0) + \Phi_{12}(k, 0) \Pi_{0}$ is invertible. 
In this case,
the solution of \eqref{eqn:pi-Q} is unique and is given by
\begin{equation} \label{sol:pi-Q-0}
\Pi_{k} = \! \Big(\Phi_{21}(k, 0) + \Phi_{22}(k, 0) \Pi_{0}\Big) \! \Big(\Phi_{11}(k, 0) + \Phi_{12}(k, 0) \Pi_{0}\Big)^{-1} \! .
\end{equation}
Moreover, for all $\ell = 0, 1, \ldots, N$, 
\begin{equation} \label{sol:pi-Q-s}
\Pi_{k} = \! \Big(\Phi_{21}(k, \ell) + \Phi_{22}(k, \ell) \Pi_{\ell}\Big) \! \Big(\Phi_{11}(k, \ell) + \Phi_{12}(k, \ell) \Pi_{\ell}\Big)^{-1}.
\end{equation}
\end{proposition}

\begin{proof}
The proof of Proposition~\ref{prp:pi-Q-exist} is given in Appendix~B. \hfill 
\end{proof}




A number of necessary and sufficient conditions for the solution of the matrix Riccati difference equation \eqref{eqn:pi-Q} to exist and to satisfy Property~\ref{pty:pi-pos} are summarized below. 
Before proceeding, and
for the sake of notational simplicity, define, for all $k = 0, 1, \ldots, N-1$, the matrices 
$L_{k} \triangleq \Phi_{11}(k, 0)^{-1} A_{k}^{-1} B_{k}$, 
$T_{k} \triangleq I_{p} + B_{k}\t A_{k}^{- \mbox{\tiny\sf T}} \big(Q_{k} - \Phi_{21}(k, 0) \Phi_{11}(k, 0)^{-1}\big) A_{k}^{-1} B_{k} \succ 0$, 
and 
\begin{align*} 
&U_{k} \triangleq 
\\* 
&\hspace{-1mm}
\begin{bmatrix}
T_{k} &    0    &    0   &   0   \\
    0   & T_{k-1} &    0   &   0   \\
    0   &    0    & \ddots &   0   \\
    0   &    0    &    0   & T_{0} 
\end{bmatrix}
\hspace{-0.6mm} - \hspace{-0.6mm}
\begin{bmatrix}
L_{k}\t \\
L_{k-1}\t \\
\vdots \\ 
L_{0}\t
\end{bmatrix}
\Pi_{0} 
\begin{bmatrix}
L_{k} & L_{k-1} & \cdots & L_{0}
\end{bmatrix} \hspace{-1mm} .
\end{align*}


\begin{theorem} \label{thm:pi-pos-exist-cond}
Under Assumption~\ref{asm:A-inv}, the following statements are equivalent regarding the solution $\Pi_{k}$ of \eqref{eqn:pi-Q}.
\begin{enumerate}[label=\roman*)]
\item For all $k = 0, 1, \ldots, N-1$, $I_{p} + B_{k}\t \Pi_{k+1} B_{k} \succ 0$ 
(respectively, for all $k = 0, 1, \ldots, N$, $\Pi_{k}$ exists). 
\item For all $k = 0, 1, \ldots, N-1$, $I_{p} - B_{k}\t A_{k}^{- \mbox{\tiny\sf T}} \big(\Pi_{k} - Q_{k}\big) A_{k}^{-1} B_{k} \succ 0$ 
(respectively, is invertible). 
\item For all $k = 1, 2, \ldots, N-1$, $U_{0}, U_{k}/U_{k-1} \succ 0$ 
(respectively, are invertible), where $U_{k}/U_{k-1}$ denotes the Schur complement of the block $U_{k-1}$ in the matrix $U_{k}$. 
\item For all $k = 0, 1, \ldots, N-1$, $U_{k} \succ 0$ 
(respectively, is invertible). 
\item $U_{N-1} \succ 0$ 
(respectively, is invertible).
\item For all $k = 0, 1, \ldots, N$, the eigenvalues of $I_{n} + \Phi_{11}(k, 0)^{-1} \Phi_{12}(k, 0) \Pi_{0}$ are positive 
(respectively, are nonzero). 
\item All eigenvalues of $I_{n} + \Phi_{11}(N, 0)^{-1} \Phi_{12}(N, 0) \Pi_{0}$ are positive 
(respectively, are nonzero). 
\end{enumerate}
Under Assumptions~\ref{asm:A-inv} and \ref{asm:G-inv}, any of the above statements is also equivalent to the statement that  
\begin{equation*}
\Pi_{0} + \Phi_{12}(N, 0)^{-1} \Phi_{11}(N, 0) \prec 0 
~~\text{(respectively, is invertible)}. 
\end{equation*}
\end{theorem}

\begin{proof}
The proof of Theorem~\ref{thm:pi-pos-exist-cond} is given in Appendix~D. \hfill 
\end{proof}

\subsection{Existence and Uniqueness of the Solution} \label{subsec:exist-unique}

In this section, we show that the solution of the coupled matrix difference equations \eqref{eqn:coupled} exists and it is unique. 
First, we need an alternative expression for the state transition matrix of $\bar{A}_{k}$ in Theorem~\ref{thm:exist-unique-Q}.

\begin{proposition} \label{prp:phi-AB-pi-Q} 
Under Assumption~\ref{asm:A-inv}, suppose the condition in Proposition~\ref{prp:pi-Q-exist} holds so that \eqref{eqn:pi-Q} admits a solution $\Pi_{k}$ from time $0$ to $N$. 
Then, for all $s, k = 0, 1, \ldots, N$, the state transition matrix of $\bar{A}_{k}$ is given by
\begin{equation} \label{eqn:phi-AB-pi-Q}
\Phi_{\bar{A}}(k, s) = \Phi_{11}(k, s) + \Phi_{12}(k, s) \Pi_{s}. 
\end{equation}
\end{proposition}

\begin{proof}
When $s = k$, we can check that 
$\Phi_{\bar{A}}(s, s) = I_{n} = \Phi_{11}(s, s) + \Phi_{12}(s, s) \Pi(s)$. 
It is easy to verify that $\Phi_{\bar{A}}(k, s) = \bar{A}_{k-1} \Phi_{\bar{A}}(k-1, s)$, regardless of $k > s$ or $k < s$. 
From \eqref{eqn:phi-M-blk} and \eqref{sol:pi-Q-s}, we have
$
\Phi_{11}(k, s) + \Phi_{12}(k, s) \Pi_{s} = 
\bar{A}_{k-1} \big(\Phi_{11}(k-1, s) + \Phi_{12}(k-1, s) \Pi_{s}\big)
$. \hfill
\end{proof}


The solution of the covariance equation \eqref{eqn:sigma-Q} is given as follows.

\begin{proposition} \label{prp:sigma1-Q}
Let Assumptions~\ref{asm:A-inv} and \ref{asm:G-inv} hold, and 
suppose $\Pi_{0} \prec - \Phi_{12}(N, 0)^{-1} \Phi_{11}(N, 0)$. 
Let $\Sigma_{0} = \bar{\Sigma}_{0} \succ 0$.
Then, the solution $\Sigma_{k}$ of \eqref{eqn:sigma-Q} from time $0$ to $N$ is given by, 
\begin{align} \label{sol:sigma-Q}
\Sigma_{k+1} &= \Phi_{\bar{A}}(k+1, 0) \bar{\Sigma}_{0} \Phi_{\bar{A}}(k+1, 0)\t 
\nonumber \\* 
&\hspace{2mm} + \sum_{i=0}^{k} \Phi_{\bar{A}}(k+1, i+1) D_{i} D_{i}\t \Phi_{\bar{A}}(k+1, i+1)\t,
\end{align}
where $\Phi_{\bar{A}}(k, s)$ is given by \eqref{eqn:phi-AB-pi-Q} and $\Pi_{k}$ is given by \eqref{sol:pi-Q-0}. 
For all $\ell = 1, 2, \ldots, n$, if $\ell$ eigenvalues of $\Pi_{0}$ go to $-\infty$, then, $\ell$ eigenvalues of $\Sigma_{N}$ approach $+\infty$. 
Furthermore, when Assumption~\ref{asm:G-D} also holds, if $\ell$ eigenvalues of $- \Phi_{12}(N, 0)^{-1} \Phi_{11}(N, 0) - \Pi_{0}$ go to $0$, then, $\ell$ eigenvalues of $\Sigma_{N} - D_{N-1} D_{N-1}\t$ approach $0$. 
\end{proposition}

\begin{proof}
Clearly, \eqref{sol:sigma-Q} is the solution of \eqref{eqn:sigma-Q}. 
It follows from Corollary~\ref{cor:phi-12-inv} and \cite[Lemma~4]{liu2024reach} that if $\ell$ eigenvalues of $\Pi_{0}$ go to $-\infty$, then, $\ell$ eigenvalues of $I_{n} + \Phi_{11}(N, 0)^{-1} \Phi_{12}(N, 0) \Pi_{0}$ approach $+\infty$. 
Since $\Phi_{\bar{A}}(N, 0) = \Phi_{11}(N, 0) \big[I_{n} + \Phi_{11}(N, 0)^{-1} \Phi_{12}(N, 0) \Pi_{0}\big]$, from the invertibility of $\Phi_{11}(N, 0)$ and the fact that $\bar{\Sigma}_{0} \succ 0$ it follows that $\ell$ eigenvalues of $\Phi_{\bar{A}}(N, 0) \bar{\Sigma}_{0} \Phi_{\bar{A}}(N, 0)\t$ approach $+\infty$. 
In light of \eqref{sol:sigma-Q}, $\ell$ eigenvalues of $\Sigma_{N}$ approach $+\infty$.

Let Assumption~\ref{asm:G-D} hold. 
Let $\bar{P} \preceq - \Phi_{12}(N, 0)^{-1} \Phi_{11}(N, 0)$ such that $\ell$ eigenvalues of $- \Phi_{12}(N, 0)^{-1} \Phi_{11}(N, 0) - \bar{P}$ are $0$. 
Since $\Phi_{11}(N, 0)$ is invertible, it suffices to show that, as $\Pi_{0} \to \bar{P}$, then, $\ell$ eigenvalues of $\Phi_{11}(N, 0)^{-1} \big(\Sigma_{N} - D_{N-1} D_{N-1}\t\big) \Phi_{11}(N, 0)^{- \mbox{\tiny\sf T}}$ approach $0$. 
In view of \eqref{sol:sigma-Q} and the fact that $\bar{\Sigma}_{0} \succ 0$, if $\ell$ eigenvalues of $\Phi_{11}(N, 0)^{-1} \big(\Sigma_{N} - D_{N-1} D_{N-1}\t\big) \Phi_{11}(N, 0)^{- \mbox{\tiny\sf T}}$ approach $0$, then, $\ell$ eigenvalues of $I_{n} + \Phi_{11}(N, 0)^{-1} \Phi_{12}(N, 0) \Pi_{0}$ go to $0$. 
In light of \cite[Lemma~4]{liu2024reach}, only the positive eigenvalues of $\Pi_{0}$ can possibly make $I_{n} + \Phi_{11}(N, 0)^{-1} \Phi_{12}(N, 0) \Pi_{0}$ singular. 
Hence, it suffices to consider the case when $\bar{P} \succ 0$. 
The case when $\bar{P}$ is not positive definite can be reduced to the case when $\bar{P} \succ 0$. 
Since $- \Phi_{12}(N, 0)^{-1} \Phi_{11}(N, 0) \succeq \bar{P} \succ 0$ and $\ell$ eigenvalues of $- \Phi_{12}(N, 0)^{-1} \Phi_{11}(N, 0) - \bar{P} \succeq 0$ are $0$, it follows from Woodbury matrix identity~\cite{horn2012matrix} that $\ell$ eigenvalues of $\bar{P}^{-1} - \big(- \Phi_{11}(N, 0)^{-1} \Phi_{12}(N, 0)\big) \succeq 0$ are $0$. 
Hence, we can write $\bar{P} = \big(- \Phi_{11}(N, 0)^{-1} \Phi_{12}(N, 0) + \Gamma \Gamma\t\big)^{-1}$, where $\Gamma \in \R^{n \times (n-\ell)}$ and $\rank \Gamma = n - \ell$. 
In view of the Woodbury identity, as $\Pi_{0} \to \bar{P}$, then, $I_{n} + \Phi_{11}(N, 0)^{-1} \Phi_{12}(N, 0) \Pi_{0} \to - \Gamma \big(I_{n-\ell} - \Gamma\t \Phi_{12}(N, 0)^{-1} \Phi_{11}(N, 0) \Gamma\big)^{-1} \Gamma\t \Phi_{12}(N, 0)^{-1} \Phi_{11}(N, 0)$. 
Let $z \in \ker (\Gamma \Gamma\t)$ be a nonzero vector, then, $z\t \Gamma = 0$. 
It follows that, as $\Pi_{0} \to \bar{P}$, then, $I_{n} + \Phi_{11}(N, 0)^{-1} \Phi_{12}(N, 0) \Pi_{0}$ approaches some matrix for which $z$ is in its left null space.

Next, we show that as $\Pi_{0} \to \bar{P}$, then, for each $i = 0, 1, \ldots, N-2$, $\Phi_{11}(N, 0)^{-1} \Phi_{\bar{A}}(N, i+1) D_{i}$ approaches some matrix for which $z$ is in its left null space as well. 
Let $i = 0, 1, \ldots, N-2$ be fixed. 
From Assumption~\ref{asm:G-D}, \eqref{eqn:phi-mono}, and \eqref{eqn:phi-11-img}, it follows that $\range D_{i} \subseteq \range \Phi_{11}(N, i+1)^{-1} \Phi_{12}(N, i+1)$. 
Since $\Phi_{11}$ is invertible, it suffices to consider the left null space of $\Phi_{11}(N, 0)^{-1} \Phi_{\bar{A}}(N, i+1) \Phi_{11}(N, i+1)^{-1} \Phi_{12}(N, i+1) \Phi_{11}(N, i+1)\t \Phi_{11}(N, 0)^{- \mbox{\tiny\sf T}} = \Phi_{11}(N, 0)^{-1} \big[\Phi_{11}(N, i+1) \Phi_{12}(N, i+1)\t + \Phi_{12}(N, i+1) \Pi_{i+1} \Phi_{12}(N, i+1)\t\big] \Phi_{11}(N, 0)^{- \mbox{\tiny\sf T}}$. 
From \eqref{eqn:phi-M-blk} and Lemma~\ref{lem:phi-eqn} it follows that $- \Phi_{11}(N, 0)^{-1} \Phi_{12}(N, 0) + \Phi_{11}(i+1, 0)^{-1} \Phi_{12}(i+1, 0) = - \Phi_{11}(i+1, 0)^{-1} \Phi_{12}(N, i+1)\t \Phi_{11}(N, 0)^{- \mbox{\tiny\sf T}} \triangleq W \succeq 0$. 
In view of \eqref{eqn:pi-k-alt}, it follows that as $\Pi_{0} \to \bar{P}$, then, $\Phi_{11}(N, 0)^{-1} \big[\Phi_{11}(N, i+1) \Phi_{12}(N, i+1)\t + \Phi_{12}(N, i+1) \Pi_{i+1} \Phi_{12}(N, i+1)\t\big] \Phi_{11}(N, 0)^{- \mbox{\tiny\sf T}} \to W (\Gamma \Gamma\t + W)^{+} W - W$, where $M^{+}$ denotes the pseudo-inverse of the matrix $M$. 
Thus, it suffices to show that $z\t \big[W (\Gamma \Gamma\t + W)^{+} W - W\big] = 0$. 
Let $\rank W = r$ and $\rank [W ~~ \Gamma] = m$, where $r \leq m \leq n$. 
With an appropriate coordinate transformation, we can partition $W$ and $\Gamma$ as follows: 
\begin{equation*}
W_{n \times n} = 
\begin{bmatrix}
W_{1} & 0_{r \times (n-r)} \\
0_{(m-r) \times r} & 0 \\
0_{(n-m) \times r} & 0 
\end{bmatrix}
, \quad 
\Gamma_{n \times (n-\ell)} = 
\begin{bmatrix}
\Gamma_{1} \\
\Gamma_{2} \\
0
\end{bmatrix}, 
\end{equation*}
where $W_{1} \in \R^{r \times r}$, $W_{1} \succ 0$, $\Gamma_{1} \in \R^{r \times (n-\ell)}$, and $\Gamma_{2} \in \R^{(m-r) \times (n-\ell)}$. 
Let $z = [z_{1}\t ~~ z_{2}\t ~~ z_{3}\t]\t \in \R^{n}$, where $z_{1} \in \R^{r}$, $z_{2} \in \R^{m-r}$, and $z_{3} \in \R^{n-m}$. 
Since $z\t \Gamma = 0$, we have $z_{1}\t \Gamma_{1} + z_{2}\t \Gamma_{2} = 0$. 
Since $\Gamma_{2} P_{2} = 0$, it follows that $z_{1}\t \Gamma_{1} P_{2} = z_{2}\t \Gamma_{2} P_{2} = 0$. 
Since 
$W - W (W + \Gamma \Gamma\t)^{+} W = 
\begin{bmatrix}
\Gamma_{1} P_{2} (I_{n-\ell} + P_{2} \Gamma_{1}\t W_{1}^{-1} \Gamma_{1} P_{2})^{-1} P_{2} \Gamma_{1}\t & 0_{r \times (n-r)} \\
0_{(n-r) \times r} & 0_{(n-r) \times (n-r)}
\end{bmatrix}$, 
where $P_{2} \triangleq I_{n-\ell} - \Gamma_{2}\t (\Gamma_{2} \Gamma_{2}\t)^{-1} \Gamma_{2}$~\cite[Lemma~5]{liu2024reach}, we can verify that $z\t \big(W - W (W + \Gamma \Gamma\t)^{+} W\big) = 0$. 
Since $z$ can take $\ell$ linearly independent values, in light of \eqref{sol:sigma-Q}, $\ell$ eigenvalues of $\Sigma_{N} - D_{N-1} D_{N-1}\t$ approach $0$. \hfill
\end{proof}


Proposition~\ref{prp:sigma1-Q} can be used to provide an explicit map from $\Pi_{0}$ to $\Sigma_{N}$ for the coupled matrix difference equations \eqref{eqn:pi-Q} and \eqref{eqn:sigma-Q} with $\Sigma_{0} = \bar{\Sigma}_{0} \succ 0$. 
Specifically, from \eqref{sol:sigma-Q} we can write 
\begin{align} \label{sol:sigma-N}
\Sigma_{N} 
&= \Phi_{\bar{A}}(N, 0) \bigg[ \bar{\Sigma}_{0} + \sum_{i=0}^{N-1} \Phi_{\bar{A}}(0, i+1) D_{i} D_{i}\t \Phi_{\bar{A}}(0, i+1)\t \bigg] 
\nonumber \\* 
&\hspace{50mm} 
\times \Phi_{\bar{A}}(N, 0)\t. 
\end{align}
Let 
$\mathcal{P} \triangleq \big\{\Pi \in \mathbb{R}^{n \times n} \,|\, \Pi = \Pi\t \prec - \Phi_{12}(N, 0)^{-1} \Phi_{11}(N, 0)\big\}$ 
and 
$\mathcal{R} \triangleq \big\{\Sigma \in \mathbb{R}^{n \times n} \,|\, \Sigma = \Sigma\t \succ D_{N-1} D_{N-1}\t\big\}$. 
From \eqref{sol:sigma-N} and \eqref{eqn:phi-AB-pi-Q}, let $f: \mathcal{P} \to \mathcal{R}$ be defined by $\Sigma_{N} = f (\Pi_{0})$, where 
\begin{align} \label{map:pi0-sigma1-Q}
&f(\Pi_{0}) 
= \Big(\Phi_{11}(N, 0) + \Phi_{12}(N, 0) \Pi_{0}\Big) 
\bigg[ \bar{\Sigma}_{0} 
\nonumber \\* 
& \! + \hspace{-1mm} \sum_{i=0}^{N-1} \hspace{-1mm} \Big( \hspace{-0.5mm} \Phi_{11}(i \hspace{-0.5mm} + \hspace{-0.5mm} 1, 0) \hspace{-0.5mm} + \hspace{-0.5mm} \Phi_{12}(i \hspace{-0.5mm} + \hspace{-0.5mm} 1, 0) \Pi_{0} \hspace{-0.5mm} \Big)^{-1} 
D_{i} D_{i}\t 
\Big( \hspace{-0.5mm} \Phi_{11}(i \hspace{-0.5mm} + \hspace{-0.5mm} 1, 0)\t 
\nonumber \\* 
& \! + \Pi_{0} \Phi_{12}(i \hspace{-0.5mm} + \hspace{-0.5mm} 1, 0)\t \Big)^{-1} \bigg] \! 
\Big(\Phi_{11}(N, 0)\t \hspace{-0.5mm} + \hspace{-0.5mm} \Pi_{0} \Phi_{12}(N, 0)\t \Big). 
\end{align}

\begin{remark}
In light of Proposition~\ref{prp:suff-Q}, under Assumptions~\ref{asm:A-inv} and \ref{asm:G-inv}, if $\bar{\Sigma}_{N} \in \image(f) \subseteq \mathcal{R}$, then, the control $u^{*}$ given by \eqref{ctrl:opt-Q} is optimal for the covariance steering problem. 
\end{remark}

To proceed, we compute the Jacobian of the map $f$ defined by \eqref{map:pi0-sigma1-Q}. 
For notational simplicity, we will use the notation $\Phi_{ij}^{k,0} \triangleq \Phi_{ij}(k, 0)$, where $i, j \in \{1, 2\}$ and let $\Delta \Pi_{0}$ denote a small increment in $\Pi_{0}$. 
Then, from~\cite{henderson1981deriving} we can write 
\begin{align*}
&\Big( \hspace{-0.5mm} \Phi_{11}^{i+1,0} \hspace{-0.5mm} + \hspace{-0.5mm} \Phi_{12}^{i+1,0} \Pi_{0} \hspace{-0.5mm} + \hspace{-0.5mm} \Phi_{12}^{i+1,0} \Delta \Pi_{0} \hspace{-0.5mm} \Big)^{-1} 
\hspace{-2.1mm} = \hspace{-0.8mm} \Big( \hspace{-0.5mm} \Phi_{11}^{i+1,0} \hspace{-0.5mm} + \hspace{-0.5mm} \Phi_{12}^{i+1,0} \Pi_{0} \hspace{-0.5mm} \Big)^{-1} 
\\*
&\hspace{1.5mm} 
- \Big( \hspace{-0.5mm} \Phi_{11}^{i+1,0} \hspace{-0.5mm} + \hspace{-0.5mm} \Phi_{12}^{i+1,0} \Pi_{0} \hspace{-0.5mm} \Big)^{-1} \hspace{-1mm}
\Phi_{12}^{i+1,0} \Delta \Pi_{0} \Big( \hspace{-0.5mm} \Phi_{11}^{i+1,0} \hspace{-0.5mm} + \hspace{-0.5mm} \Phi_{12}^{i+1,0} \Pi_{0} \hspace{-0.5mm} \Big)^{-1} 
\\*
&\hspace{60mm} 
+ O\Big(\big\| \Delta \Pi_{0} \big\|^{2}\Big).
\end{align*}
After collecting all the first order terms of $\Delta \Pi_{0}$, we obtain
\begin{align} \label{eqn:map-diff-Q}
&f\big(\Pi_{0} + \Delta \Pi_{0}\big) \hspace{-0.5mm} - \hspace{-0.5mm} f\big(\Pi_{0}\big) \hspace{-0.5mm} = \hspace{-0.5mm} O\Big( \hspace{-0.5mm} \big\| \Delta \Pi_{0} \big\|^{2} \Big) 
\hspace{-0.5mm} + \hspace{-0.5mm} \Phi_{\bar{A}}^{N, 0} \bigg[ W_{N,0} \Delta \Pi_{0} \bar{\Sigma}_{0} 
\nonumber \\* 
&\hspace{10mm} 
+ \bar{\Sigma}_{0} \Delta \Pi_{0} W_{N,0} 
+ \sum_{i=0}^{N-1} \Big[ \big(W_{N,0} - W_{i+1,0}\big) \Delta \Pi_{0} P_{i} 
\nonumber \\* 
&\hspace{10mm} 
+ P_{i} \Delta \Pi_{0} \big(W_{N,0} - W_{i+1,0}\big) \Big] \bigg]  \big(\Phi_{\bar{A}}^{N, 0}\big)\t, 
\end{align}
where $\Phi_{\bar{A}}^{N, 0} \triangleq \Phi_{\bar{A}}(N, 0)$, 
and, for $k = 0, 1, \ldots, N$, 
$P_{i} \triangleq 
\Big(\Phi_{11}^{i+1,0} + \Phi_{12}^{i+1,0} \Pi_{0}\Big)^{-1} 
D_{i} D_{i}\t 
\Big(\big(\Phi_{11}^{i+1,0}\big)\t + \Pi_{0} \big(\Phi_{12}^{i+1,0}\big)\t \Big)^{-1} \! \succeq 0$, 
and
\begin{align} \label{def:W}
&W_{k,0} \triangleq 
\\* 
&\Big( \hspace{-0.5mm} I_{n} + \Phi_{11}(k, 0)^{-1} \Phi_{12}(k, 0) \Pi_{0} \hspace{-0.5mm} \Big)^{-1} \Phi_{11}(k, 0)^{-1} \Phi_{12}(k, 0) \preceq 0. \nonumber
\end{align}


Given an $n \times n$ matrix $H = [h_{ij}]$, its vectorized version is 
$\vect(H) \triangleq [h_{11} \enspace \ldots \enspace h_{n1} \enspace h_{12} \enspace \ldots \enspace h_{n2} \enspace \ldots \enspace h_{1n} \enspace \ldots \enspace h_{nn}]\t$. 
Define the map $\bar{f}: \big\{\vect(\Pi_{0}) \in \mathbb{R}^{n^{2}} \,\big|\, \Pi_{0} \in \mathcal{P}\big\} \to \big\{\vect(\Sigma_{N}) \in \mathbb{R}^{n^{2}} \,\big|\, \Sigma_{N} \in \mathcal{R}\big\}$ such that 
$
\bar{f}\big(\vect(\Pi_{0})\big) = \vect\big(f(\Pi_{0})\big)
$, 
where $f$ is defined in \eqref{map:pi0-sigma1-Q}. 
It follows from vectorizing both sides of \eqref{eqn:map-diff-Q} that 
\begin{multline*}
\bar{f}\big(\vect(\Pi_{0}) + \vect(\Delta \Pi_{0})\big) - \bar{f}\big(\vect(\Pi_{0})\big) = 
\\* 
\partial \bar{f}\big(\vect(\Pi_{0})\big) \vect\big(\Delta \Pi_{0}\big) 
+ O\Big(\big\| \Delta \Pi_{0} \big\|^{2}\Big),
\end{multline*}
where, 
\begin{multline} \label{eqn:map-jacob-Q}
\partial \bar{f}\big(\vect(\Pi_{0})\big) = 
\Phi_{\bar{A}}^{N, 0} \otimes \Phi_{\bar{A}}^{N, 0} \bigg[ \bar{\Sigma}_{0} \otimes W_{N,0} + W_{N,0} \otimes \bar{\Sigma}_{0} 
\\* 
+ \sum_{i=0}^{N-1} \Big[ P_{i} \otimes \big(W_{N,0} - W_{i+1,0}\big) + \big(W_{N,0} - W_{i+1,0}\big) \otimes P_{i} \Big] \bigg],
\end{multline}
where $\otimes$ denotes the Kronecker product. 
Thus, $\partial \bar{f}\big(\vect(\Pi_{0})\big)$ is the Jacobian of the map $\bar{f}$ at $\vect(\Pi_{0})$.


Finally, we are ready to show the existence and uniqueness of the solution to the coupled matrix difference equations \eqref{eqn:coupled}.

\begin{proposition} \label{prp:map-inv-Q}
Let Assumptions~\ref{asm:A-inv}-\ref{asm:G-D} hold. 
Then, for any given $\bar{\Sigma}_{0} \succ 0$, the map $f$ defined by \eqref{map:pi0-sigma1-Q} is a homeomorphism. 
Thus, for any $\Sigma_{N} \in \mathcal{R}$, there exists a unique $\Pi_{0} \in \mathcal{P}$ such that $\Sigma_{N} = f(\Pi_{0})$.
\end{proposition}

\begin{proof}
Clearly, $\partial \bar{f}\big(\vect(\Pi_{0})\big)$ is continuous in $\vect(\Pi_{0})$. 
First, we show that, for each $\Pi_{0} \in \mathcal{P}$, $\partial \bar{f}\big(\vect(\Pi_{0})\big)$ is nonsingular. 
Since $\Phi_{\bar{A}}^{N, 0} \otimes \Phi_{\bar{A}}^{N, 0}$ is nonsingular, it suffices to show that the term in the large square brackets of \eqref{eqn:map-jacob-Q}, that is, 
$
S \triangleq 
\bar{\Sigma}_{0} \otimes W_{N,0} + W_{N,0} \otimes \bar{\Sigma}_{0} 
\hspace{-0.5mm} + \hspace{-1mm} \sum_{i=0}^{N-1} \hspace{-0.5mm} \Big[ P_{i} \otimes \big(W_{N,0} - W_{i+1,0}\big)  
+ \big(W_{N,0} - W_{i+1,0}\big) \otimes P_{i} \Big]
$ 
is nonsingular. 
From Lemma~\ref{lem:phi-eqn} and Corollary~\ref{cor:phi-12-inv} in Appendix~C, we have that $\Phi_{11}(N, 0)$ and $\Phi_{12}(N, 0)$ are invertible.
It then follows from \eqref{def:W} that $W_{N,0} = \big(\Phi_{12}(N, 0)^{-1} \Phi_{11}(N, 0) + \Pi_{0}\big)^{-1}$. 
Since $\Pi_{0} \in \mathcal{P}$, we have $W_{N,0} \prec 0$. 
From Lemma~\ref{lem:phi-11-12-inv-mono} in Appendix~E, it follows that, for all $i = 0, 1, \ldots, N-1$, $W_{N,0} - W_{i+1,0} \preceq 0$. 
One can readily check that $S$ is symmetric, because $\bar{\Sigma}_{0} \succ 0$, $W_{N,0} \prec 0$, $P_{i} \succeq 0$, and $W_{N,0} - W_{i+1,0} \preceq 0$ are all symmetric. 
Let $X \neq 0$ be an $n \times n$ matrix. 
Then,
\begin{align*}
&\vect(X)\t S \vect(X) 
\\
&= \trace \bigg( X\t W_{N,0} X \bar{\Sigma}_{0} + X\t \bar{\Sigma}_{0} X W_{N,0} + \sum_{i=0}^{N-1} \Big[ X\t 
\\*
&\hspace{5mm} \times 
\big(W_{N,0} - W_{i+1,0}\big) X P_{i} + X\t P_{i} X \big(W_{N,0} - W_{i+1,0}\big) \Big] \bigg) 
\\
&\leq \trace \Big( \bar{\Sigma}_{0}^{\frac{1}{2}} X\t W_{N,0} X \bar{\Sigma}_{0}^{\frac{1}{2}} + \bar{\Sigma}_{0}^{\frac{1}{2}} X W_{N,0} X\t \bar{\Sigma}_{0}^{\frac{1}{2}} \Big) < 0.
\end{align*}
Thus, $S \prec 0$. 
Therefore, $\partial \bar{f}\big(\vect(\Pi_{0})\big)$ is nonsingular at each $\vect(\Pi_{0})$ in the domain of $\bar{f}$.

Next, we show that $f$ is proper, that is, for any compact subset $\mathcal{K} \subseteq \mathcal{R}$, 
the inverse image $f^{-1}(\mathcal{K}) \subseteq \mathcal{P}$ is compact. 
Since $\mathcal{K}$ is bounded in $\mathbb{R}^{n \times n}$, in view of \eqref{map:pi0-sigma1-Q}, the set
\begin{equation*}
\hspace{-0.5mm} \Big\{ \hspace{-0.6mm} \Big( \hspace{-0.5mm} \Phi_{11}^{N,0} + \Phi_{12}^{N,0} \Pi_{0} \hspace{-0.5mm} \Big) \bar{\Sigma}_{0} \Big( \hspace{-0.5mm} \big(\Phi_{11}^{N,0}\big)\t + \Pi_{0} \big(\Phi_{12}^{N,0}\big)\t \Big) \Big| \Pi_{0} \in f^{-1}(\mathcal{K}) \hspace{-0.5mm} \Big\}
\end{equation*}
is also bounded in $\mathbb{R}^{n \times n}$. 
Since $\bar{\Sigma}_{0}$ and $\Phi_{12}^{N,0}$ are invertible, $f^{-1}(\mathcal{K})$ is bounded in $\mathbb{R}^{n \times n}$. 
In light of Proposition~\ref{prp:sigma1-Q}, as $\Pi_{0}$ approaches the boundary of $\mathcal{P}$, then, $\Sigma_{N}$ approaches the boundary of $\mathcal{R}$. 
Since $f$ is continuous and $\mathcal{K}$ is closed in $\mathbb{R}^{n \times n}$, the inverse image $f^{-1}(\mathcal{K})$ is also closed in $\mathbb{R}^{n \times n}$. 
Therefore, $f^{-1}(\mathcal{K})$ is compact, and thus $f$ is proper. 
Since the set $\mathcal{R}$ is convex, it is simply connected~\cite{krantz2002implicit}. 
From Hadamard's global inverse function theorem~\cite{krantz2002implicit}, $f$ is a homeomorphism. \hfill
\end{proof}

\subsection{Separation of Mean and Covariance Steering} \label{subsec:separation}

In this section, we establish the independence of the optimal mean steering and covariance steering problems.

\begin{theorem} \label{thm:mean-covar}
Under Assumptions~\ref{asm:A-inv} and \ref{asm:G-inv}, 
the optimal control \eqref{ctrl:opt-Q} can be written in the form 
\begin{equation} \label{eq:state-FB}
u_{k}^{*} = K_{k} \big(x_{k}^{*} - \mu_{k}\big) + v_{k},
\end{equation}
where $K_{k} = - \big(R_{k} + B_{k}\t \Pi_{k+1} B_{k}\big)^{-1} B_{k}\t \Pi_{k+1} A_{k}$ is the state feedback matrix, 
$\mu_{k} \triangleq \mathbb{E}[x_{k}^{*}]$ is the optimal state mean given by 
\begin{multline} \label{eqn:mu-k}
\mu_{k} = \Big(\Phi_{11}(k, 0) - \Phi_{12}(k, 0) \Phi_{12}(N, 0)^{-1} \Phi_{11}(N, 0)\Big) \bar{\mu}_{0} 
\\* 
+ \Phi_{12}(k, 0) \Phi_{12}(N, 0)^{-1} \bar{\mu}_{N}, 
\end{multline}
and $v_{k}$ is the optimal feed-forward control term given by 
\begin{align} \label{eqn:v-k}
v_{k} &= - R_{k}^{-1} B_{k}\t \Big(\Phi_{21}(k+1, 0) - \Phi_{22}(k+1, 0) \Phi_{12}(N, 0)^{-1} 
\nonumber \\* 
& \times \Phi_{11}(N, 0)\Big) \bar{\mu}_{0} 
- R_{k}^{-1} B_{k}\t \Phi_{22}(k+1, 0) \Phi_{12}(N, 0)^{-1} \bar{\mu}_{N}. 
\end{align}
\end{theorem}

\begin{proof}
First, let 
\begin{equation} \label{def:B-bar}
\bar{B}_{k} \triangleq B_{k} \big(R_{k} + B_{k}\t \Pi_{k+1} B_{k} \big)^{-\frac{1}{2}}. 
\end{equation}
Let $\bar{G}(k, s)$ denote the reachability Gramian of the pair $\big\{\big(\bar{A}_{i}, \bar{B}_{i}\big) \, \big| \, i = s, s \pm 1, \ldots, k \big\}$ from time $s$ to $k$. 
From Lemma~\ref{lem:gramian-bar} in Appendix~E, it follows that $\bar{G}(k, s) = - \Phi_{12}(k, s) \Phi_{\bar{A}}(k, s)\t$. 
In light of \eqref{ctrl:opt-Q} and \eqref{eq:state-FB}, we have 
\begin{equation} \label{eqn:K-k}
K_{k} = - \big(R_{k} + B_{k}\t \Pi_{k+1} B_{k}\big)^{-1} B_{k}\t \Pi_{k+1} A_{k}. 
\end{equation}
To show \eqref{eqn:mu-k}, notice that
from \eqref{eqn:mu-Q}, \eqref{sol:psi-Q}, \eqref{eqn:gramian-bar}, and \eqref{eqn:phi-AB-pi-Q}, it follows that 
$
\mu_{k} 
= 
\big(\Phi_{11}(k, 0) - \Phi_{12}(k, 0) \Phi_{12}(N, 0)^{-1} \Phi_{11}(N, 0)\big) \bar{\mu}_{0} 
+ \Phi_{12}(k, 0) \Phi_{12}(N, 0)^{-1} \bar{\mu}_{N}
$. 
Next, we show \eqref{eqn:v-k}. 
In view of \eqref{ctrl:opt-Q}, \eqref{eq:state-FB}, \eqref{eqn:mu-k}, and \eqref{eqn:gramian-bar}, we obtain 
\begin{align*}
v_{k} \hspace{-0.5mm} 
&= 
\Big[K_{k} \Big(\Phi_{11}(k, 0) - \Phi_{12}(k, 0) \Phi_{12}(N, 0)^{-1} \Phi_{11}(N, 0)\Big) 
\\*
&\hspace{4mm} 
+ \big(R_{k} + B_{k}\t \Pi_{k+1} B_{k}\big)^{-1} B_{k}\t \Phi_{\bar{A}}(0, k+1)\t \Phi_{12}(N, 0)^{-1} 
\\*
&\hspace{4mm} 
\times \Phi_{\bar{A}}(N, 0) \Big] \bar{\mu}_{0} 
+ \Big[K_{k} \Phi_{12}(k, 0) \Phi_{12}(N, 0)^{-1} 
\\*
&\hspace{4mm} 
- \hspace{-0.5mm} \big( \hspace{-0.5mm} R_{k} \hspace{-0.5mm} + \hspace{-0.5mm} B_{k}\t \Pi_{k+1} B_{k} \hspace{-0.7mm} \big)^{-1} \hspace{-0.5mm} B_{k}\t \Phi_{\bar{A}}(0, k \hspace{-0.5mm} + \hspace{-0.5mm} 1 \hspace{-0.5mm} )\t \Phi_{12}(N, 0)^{-1} \hspace{-0.5mm} \Big] \bar{\mu}_{N}. 
\end{align*}
Since 
$\Phi_{M}(k, 0) = M_{k}^{-1} \Phi_{M}(k+1, 0)$, it follows from \eqref{eqn:M-inv}, after replacing $B_{k} B_{k}\t$ with $B_{k} R_{k}^{-1} B_{k}\t$, that 
\begin{subequations} \label{eqn:phi-k-k+1-0}
\begin{align}
\Phi_{11}(k, 0) \hspace{-0.5mm} &= \hspace{-0.5mm} A_{k}^{-1} \Phi_{11}(k \hspace{-0.5mm} + \hspace{-0.5mm} 1, 0) + A_{k}^{-1} B_{k} R_{k}^{-1} B_{k}\t \Phi_{21}(k \hspace{-0.5mm} + \hspace{-0.5mm} 1, 0), \label{eqn:phi-11-k-k+1-0}
\\ 
\Phi_{12}(k, 0) \hspace{-0.5mm} &= \hspace{-0.5mm} A_{k}^{-1} \Phi_{12}(k \hspace{-0.5mm} + \hspace{-0.5mm} 1, 0) + A_{k}^{-1} B_{k} R_{k}^{-1} B_{k}\t \Phi_{22}(k \hspace{-0.5mm} + \hspace{-0.5mm} 1, 0). \label{eqn:phi-12-k-k+1-0}
\end{align}
\end{subequations}
From \eqref{eqn:K-k}, \eqref{eqn:phi-12-k-k+1-0}, \eqref{eqn:phi-AB-pi-Q}, \eqref{eqn:phi-12-equiv}, and \eqref{eqn:phi-11-22}, it follows that the coefficient for $\bar{\mu}_{N}$ in $v_{k}$ is $- R_{k}^{-1} B_{k}\t \Phi_{22}(k+1, 0) \Phi_{12}(N, 0)^{-1}$. 
Similarly, from \eqref{eqn:K-k}, \eqref{eqn:phi-k-k+1-0}, \eqref{eqn:phi-AB-pi-Q}, \eqref{eqn:phi-12-equiv}, \eqref{eqn:phi-11-22}, and \eqref{sol:pi-Q-0}, it follows that the coefficient for $\bar{\mu}_{0}$ in $v_{k}$ is $- R_{k}^{-1} B_{k}\t \big(\Phi_{21}(k + 1, 0) - \Phi_{22}(k + 1, 0) \Phi_{12}(N, 0)^{-1} \Phi_{11}(N, 0)\big)$. 
Thus, \eqref{eqn:v-k} holds. \hfill
\end{proof}

\begin{remark}
Note that in the optimal control \eqref{eq:state-FB} the state feedback matrix $K_{k}$ depends on $\Pi_{k+1}$, which is determined solely by the initial and terminal state covariances. 
On the other hand, the state mean dynamics $\mu_{k}$ and the feed-forward control term $v_{k}$ are determined solely by the initial and terminal state means. 
\end{remark}

\section{Computation of the Optimal Control} \label{sec:compute}

In this section, we propose two numerical algorithms to solve for the optimal control law of the covariance steering problem formulated in Section~\ref{sec:prob}. 
The first algorithm adopts Newton's root finding method to compute $\Pi_{0}$ by exploiting the map $f$ given by \eqref{map:pi0-sigma1-Q} and its Jacobian, given by \eqref{eqn:map-jacob-Q}. 
The second algorithm recasts the optimal covariance steering problem as a convex optimization problem, specifically, a semi-definite program (SDP). 
It is shown that the optimal solution to the SDP problem also solves the original non-convex covariance steering problem.

\subsection{Newton's Method} \label{subsec:Newton}

We describe an approach based on Newton's root-finding method for solving the optimal covariance control problem.

The first step is to absorb $R_{k} \succ 0$ into $B_{k}$ by defining, for $k = 0, 1, \ldots, N-1$, $\hat{B}_{k} \triangleq B_{k} R_{k}^{-\frac{1}{2}}$ and $\hat{R}_{k} \equiv I_{p}$. 
For ease of notation, we will substitute $B_{k}$ and $R_{k}$ below with $\hat{B}_{k}$ and $\hat{R}_{k}$, respectively. 
The second step is, for $i, j = 1, 2$ and $k = 0, 1, \ldots, N-1$, to determine the coefficient matrices $\Phi_{ij}(k+1, 0)$ and $\Phi_{ij}(N, k)$ from \eqref{eqn:phi-M-blk}. 
The third step is to compute $\Pi_{0}$ using Newton's method. 
That is, starting from an initial guess of $\Pi_{0, 0} \in \mathcal{P}$, iterate the equation 
\begin{align*}
&\vect \big(\Pi_{0, i+1}\big) = \vect \big(\Pi_{0, i}\big) - \Big[\partial \bar{f} \big(\vect (\Pi_{0, i})\big)\Big]^{-1} 
\\* 
&\hspace{22mm} \times \Big[\bar{f} \big(\vect (\Pi_{0, i})\big) - \vect \big(\bar{\Sigma}_{N}\big)\Big], ~ i = 0, 1, \ldots, 
\end{align*}
where $\Pi_{0, i}$ denotes the value of $\Pi_{0}$ at the $i$th iteration of the Newton's method. 
This iteration should be terminated when $\big\Vert \bar{f} \big(\vect (\Pi_{0, i})\big) - \vect \big(\bar{\Sigma}_{N}\big) \big\Vert$ is sufficiently small. 
The fourth step is, for $k = 0, 1, \ldots, N-1$, to calculate $\Pi_{k+1}$ via \eqref{sol:pi-Q-0} and $\Phi_{\bar{A}}(N, k)$ via 
\eqref{eqn:phi-AB-pi-Q}. 
Finally, the optimal control $u^{*}_{k}$ can be computed using \eqref{ctrl:opt-Q}. 


\subsection{Semi-Definite Program} \label{subsec:SDP}

In this section, we develop an SDP-based formulation for solving the optimal covariance control problem. 
Our approach extends the results of~\cite{balci2022exact} by showing that the calculated control law is globally optimal among the class of all admissible, linear or nonlinear feedback control laws. 
First, it follows from Theorem \ref{thm:exist-unique-Q} that, with Assumptions \ref{asm:A-inv}-\ref{asm:G-D}, the optimal control takes the state feedback form. 

The next proposition states that without Assumptions \ref{asm:A-inv}-\ref{asm:G-D} the optimal control that solves the covariance steering problem takes the form of a randomized state feedback control.
This result, of independent interest, illustrates the fact that the solution of the optimal covariance steering problem belongs to the class of randomized state-feedback controls.
In case the additional Assumptions \ref{asm:A-inv}-\ref{asm:G-D} hold, this randomization can be injected implicitly by the state $x_{k}$.

\begin{proposition} \label{prp:covar-ctrl}
For any admissible control $\hat{u}$ of system \eqref{sys:disc-time}, there exists a control $u$ of the form 
\begin{equation} \label{eqn:ctrl-equiv}
u_{k} = K_{k} (x_{k} - \mu_{k}) + v_{k} + \nu_{k}, 
\end{equation}
where $K_{k} \in \R^{p \times n}$ is the state feedback matrix, $v_{k} \in \R^{p}$ is the feed-forward term, $\mu_{k} = \mathbb{E}[x_{k}]$, and $\nu_{k} \in \R^{p}$ is an independent, square-integrable random vector with $\mathbb{E}\left[\nu_{k}\right] = 0$ and $\mathbb{E}\left[\nu_{k} \nu_{k}\t\right] \succeq 0$, 
such that, for all $k = 1, 2, \ldots, N$, the covariances of the state $x_{k}$ under $\hat{u}$ and $u$ are the same and have the same cost \eqref{cost-Q}. 
\end{proposition}

\begin{proof}
Let the state covariance of system \eqref{sys:disc-time} under the control $\hat{u}$ be denoted by $\hat{\Sigma}_{k}$. 
Let $\hat{v}_{k} \triangleq \mathbb{E}[\hat{u}_{k}]$. 
Notice from \eqref{sys:disc-time} that, for all $k = 0, 1, \ldots, N-1$, 
\begin{align} \label{eqn:sigma-ctrl}
&\hat{\Sigma}_{k+1} 
= 
\mathbb{E}\Big[\big(A_{k} (x_{k} - \mu_{k}) + B_{k} (\hat{u}_{k} - \hat{v}_{k})\big) 
\nonumber \\*
&\hspace{15mm} \times 
\big(A_{k} (x_{k} - \mu_{k}) + B_{k} (\hat{u}_{k} - \hat{v}_{k})\big)\t\Big] + D_{k} D_{k}\t 
\nonumber \\
&= 
A_{k} \hat{\Sigma}_{k} A_{k}\t \hspace{-0.5mm} + \hspace{-0.5mm} B_{k} \hat{\Sigma}^{ux}_{k} A_{k}\t \hspace{-0.5mm} + \hspace{-0.5mm} A_{k} \hat{\Sigma}^{ux \mbox{\tiny\sf T}}_{k} B_{k}\t \hspace{-0.5mm} + \hspace{-0.5mm} B_{k} \hat{\Sigma}^{uu}_{k} B_{k}\t \hspace{-0.5mm} + \hspace{-0.5mm} D_{k} D_{k}\t, 
\end{align}
where $\hat{\Sigma}^{ux}_{k} \triangleq \mathbb{E}\big[(\hat{u}_{k} - \hat{v}_{k}) (x_{k} - \mu_{k})\t\big] \in \R^{p \times n}$ and $\hat{\Sigma}^{uu}_{k} \triangleq \mathbb{E}\big[(\hat{u}_{k} - \hat{v}_{k}) (\hat{u}_{k} - \hat{v}_{k})\t\big] \in \R^{p \times p}$ are treated as the control in \eqref{eqn:sigma-ctrl}. 
Let $z \in \ker \hat{\Sigma}_{k}$. 
Then, it follows that $(x_{k}-\mu_{k})\t z = 0$ almost surely. 
It follows immediately that $\hat{\Sigma}^{ux}_{k} z = 0$. 
Thus, $\ker \hat{\Sigma}_{k} \subseteq \ker \hat{\Sigma}^{ux}_{k}$. 
Therefore, there exists $K_{k} \in \R^{p \times n}$ such that $\hat{\Sigma}^{ux}_{k} = K_{k} \hat{\Sigma}_{k}$.
Since 
$\begin{bmatrix}
\hat{\Sigma}_{k} & \hat{\Sigma}^{ux \mbox{\tiny\sf T}}_{k} \\
\hat{\Sigma}^{ux}_{k} & \hat{\Sigma}^{uu}_{k}
\end{bmatrix}
\succeq 0$, 
$\hat{\Sigma}_{k} \succeq 0$, and $\hat{\Sigma}^{ux}_{k} = K_{k} \hat{\Sigma}_{k}$, 
it follows that $\hat{\Sigma}^{uu}_{k} \succeq K_{k} \hat{\Sigma}_{k} K_{k}\t$~\cite[Theorem~1.20]{zhang2005schur}. 
Hence, there exists $V_{k} \succeq 0$ such that $\hat{\Sigma}^{uu}_{k} = K_{k} \hat{\Sigma}_{k} K_{k}\t + V_{k}$. 
Then, \eqref{eqn:sigma-ctrl} can be written, equivalently, as 
\begin{equation} \label{eqn:sigma-ctrl-equiv}
\hat{\Sigma}_{k+1} 
= 
\big(A_{k} + B_{k} K_{k}\big) \hat{\Sigma}_{k} \big(A_{k} + B_{k} K_{k}\big)\t + B_{k} V_{k} B_{k}\t + D_{k} D_{k}\t,
\end{equation}
where $K_{k}$ and $V_{k}$ are the control variables of \eqref{eqn:sigma-ctrl-equiv}. 
Clearly, the covariance equation \eqref{eqn:sigma-ctrl-equiv} can be achieved by a control of the form \eqref{eqn:ctrl-equiv}, where $v_{k} = \hat{v}_{k} = \mathbb{E}[u_{k}]$ and $\mathbb{E}\left[\nu_{k} \nu_{k}\t\right] = V_{k} \succeq 0$. 
Notice that $\mathbb{E}\big[(u_{k} - v_{k}) (x_{k} - \mu_{k})\t\big] = K_{k} \hat{\Sigma}_{k} = \hat{\Sigma}^{ux}_{k}$ and $\mathbb{E}\big[(u_{k} - v_{k}) (u_{k} - v_{k})\t\big] = K_{k} \hat{\Sigma}_{k} K_{k}\t + V_{k} = \hat{\Sigma}^{uu}_{k}$.

Furthermore, the cost function \eqref{cost-Q} for $\hat{u}$ is 
\begin{align*}
&\hat{J} = \sum_{k=0}^{N-1} \trace\Big(Q_{k} \mathbb{E}\big[x_{k} x_{k}\t\big]\Big) + \trace\Big(R_{k} \mathbb{E}\big[\hat{u}_{k} \hat{u}_{k}\t\big]\Big)
\\
&= \sum_{k=0}^{N-1} \trace\Big( \hspace{-0.5mm} Q_{k} \big(\hat{\Sigma}_{k} + \mu_{k} \mu_{k}\t\big) \hspace{-0.5mm} \Big) 
+ \trace\Big( \hspace{-0.5mm} R_{k} \big(\hat{\Sigma}^{uu}_{k} + \hat{v}_{k} \hat{v}_{k}\t\big) \hspace{-0.5mm} \Big). 
\end{align*}
From \eqref{eqn:sigma-ctrl}, $\hat{\Sigma}$ is determined by $\hat{\Sigma}^{ux}$ and $\hat{\Sigma}^{uu}$. 
Since $\mu_{k+1} = A_{k} \mu_{k} + B_{k} \hat{v}_{k}$, $\mu$ is determined by $\hat{v}$. 
As $u$ can achieve the same $\hat{\Sigma}^{ux}$, $\hat{\Sigma}^{uu}$, and $\hat{v}$ as $\hat{u}$, their costs are the same. \hfill
\end{proof}

Let $\mu_{k} = \mathbb{E}[x_{k}]$, $\Sigma_{k} = \mathbb{E}\big[(x_{k} - \mu_{k}) (x_{k} - \mu_{k})\t\big]$, $v_{k} = \mathbb{E}[u_{k}]$, $Y_{k} = \mathbb{E}\big[(u_{k} - v_{k}) (u_{k} - v_{k})\t\big]$, and $U_{k} = \mathbb{E}\big[(u_{k} - v_{k}) (x_{k} - \mu_{k})\t\big]$. 
Using the properties of the trace and the expectation operator~\cite{bakolas2018finite}, along with the covariance propagation equation \eqref{eqn:sigma-ctrl}, the optimal covariance steering problem can be written in the equivalent form 
\begin{subequations} \label{opt:mean-covar}
\begin{align}
&\min_{\Sigma_k, U_k, Y_k, \mu_k, v_k} ~ J = \hspace{-1mm} \sum_{k = 0}^{N-1} \trace\big(Q_k \Sigma_k \big) + \trace \big( R_k Y_{k} \big) 
\nonumber \\* 
&\hspace{37mm} 
+ \mu_k\t Q_k \mu_k + v_k\t R_k v_k, \label{opt:mean-covar_cost} 
\\
&\textrm{such that, for all}\; k = 0, 1, \ldots, N-1, 
\nonumber \\
&\hspace{15mm} 
\mu_{k+1} = A_k \mu_k + B_k v_k, \quad \label{opt:mean-covar_mu} 
\\
&\hspace{15mm} 
A_k \Sigma_k A_k\t + B_k U_k A_k\t + A_k U_k\t B_k\t + D_k D_k\t 
\nonumber \\* 
&\hspace{25mm}
+ B_k Y_k B_k\t - \Sigma_{k+1} = 0, \label{opt:mean-covar_sigma}
\\
&\hspace{15mm}
\begin{bmatrix} 
\Sigma_k & U_k\t \\
U_k & Y_k
\end{bmatrix} \succeq 0, 
\label{opt:mean-covar_V}
\\
&\hspace{15mm}
\mu_0 = \bar{\mu}_0, \quad \mu_N = \bar{\mu}_N, 
\label{opt:mean-covar_meanbdr}
\\
&\hspace{15mm}
\Sigma_0 = \bar{\Sigma}_0, \quad \Sigma_N = \bar{\Sigma}_N. 
\label{opt:mean-covar_covrbdr}
\end{align}
\end{subequations}
The optimization variables of \eqref{opt:mean-covar} are $\Sigma_1$, $\Sigma_2$, $\ldots$, $\Sigma_{N-1}$, $U_0$, $U_1$, $\ldots$, $U_{N-1}$, $v_0$, $v_1$, $\ldots$, $v_{N-1}$, $\mu_1$, $\mu_2$, $\ldots$, $\mu_{N-1}$, $Y_0$, $Y_1$, $\ldots$, $Y_{N-1}$, 
whereas $\mu_0$, $\mu_N$, $\Sigma_0$, and $\Sigma_N$ are given.
To this end, we have shown the following result. 

\begin{theorem} \label{thm:prob_equiv}
The optimal solution to the convex problem \eqref{opt:mean-covar}, if it exists, is the same as the optimal solution to the covariance steering problem formulated at the end of Section \ref{sec:prob}. 
If either optimal solution does not exist, neither does the other.
\end{theorem}

The cost function \eqref{opt:mean-covar_cost} can be further decomposed into $J = J_{\Sigma}(\Sigma_{k}, U_k, Y_k) + J_{\mu}(\mu_k, v_k)$, where 
\begin{equation*}
\begin{aligned}
& J_{\Sigma} \triangleq \sum_{k = 0}^{N-1} {\trace\big(Q_k \Sigma_k \big) + \trace \big(R_k Y_k\big)}, \\
& J_{\mu}    \triangleq \sum_{k = 0}^{N-1} {\mu_k\t Q_k \mu_k + v_k\t R_k v_k}.
\end{aligned}
\end{equation*}
Since there is no coupling, the two optimization problems of $J_{\mu}$ and $J_{\Sigma}$ can be treated separately. 
The optimization of $J_{\mu}$ subject to the constraint \eqref{opt:mean-covar_mu} and the boundary condition \eqref{opt:mean-covar_meanbdr} is the mean steering subproblem of \eqref{opt:mean-covar}, which is straighforward and can even be solved analytically~\cite{okamoto2018optimal}.
We therefore focus solely on the optimization of $J_{\Sigma}$ subject to \eqref{opt:mean-covar_sigma}, \eqref{opt:mean-covar_V}, and \eqref{opt:mean-covar_covrbdr}, which is an SDP and will be referred to as the covariance steering subproblem of \eqref{opt:mean-covar}.

If, for all $k = 0, 1, \ldots, N-1$, $\Sigma_{k} \succ 0$, then, the covariance steering subproblem of~\eqref{opt:mean-covar} becomes
\begin{subequations} \label{opt:convex}
\begin{align}
&\min_{\Sigma_k, U_k, Y_k} ~ J_{\Sigma} = \sum_{k = 0}^{N-1} {\trace \big(Q_k \Sigma_k \big) + \trace \big(R_k Y_k \big)}, 
\\
&\textrm{such that, for all}\; k = 0, 1, \ldots, N-1, 
\nonumber \\
&\hspace{15mm} 
C_k \triangleq 
U_k \Sigma_k^{-1} U_k\t - Y_k \preceq 0, \label{opt:convex_relaxed} 
\\
&\hspace{15mm} 
G_k \triangleq 
A_k \Sigma_k A_k\t + B_k U_k A_k\t + A_k U_k\t B_k\t 
\nonumber \\* 
&\hspace{25mm} 
+ B_k Y_k B_k\t + D_k D_k\t - \Sigma_{k+1} = 0, \label{opt:convex_cov}
\\
&\hspace{15mm} 
\Sigma_0 = \bar{\Sigma}_0, \quad \Sigma_N = \bar{\Sigma}_N. \label{opt:convex_covrbdr}
\end{align}
\end{subequations}


\begin{remark}
We can find the solution $\Pi_{k}$ of the coupled matrix equations \eqref{eqn:coupled} as follows. 
First, from the SDP solution of \eqref{opt:convex} we obtain $U_{k}$ and $\Sigma_{k}$. 
Then, from $K_{k} = U_{k} \Sigma_{k}^{-1}$, we obtain $\bar{A}_{k} = A_{k} + B_{k} K_{k}$. 
Next, we compute the state transition matrix $\Phi_{\bar{A}}(N, 0)$. 
Lastly, since $\Phi_{11}(N, 0)$ and $\Phi_{12}(N, 0)$ are known, we can recover $\Pi_{0}$ from \eqref{eqn:phi-AB-pi-Q}. 
\end{remark}

In the rest of this section, we show that the optimal solution to the SDP problem \eqref{opt:convex} satisfies, for all $k = 0, 1, \ldots, N-1$, $C_k = 0$. 
From the definition of $U_{k}$ in~\eqref{opt:mean-covar}, there exists $K_{k} \in \R^{p \times n}$ such that $U_{k} = K_{k} \Sigma_{k}$. 
Then, it follows from $C_k = 0$ that $Y_{k} = K_{k} \Sigma_{k} K_{k}\t$. 
Hence, the optimal solution to~\eqref{opt:convex} takes the state feedback form, as expected. 
To show $C_k = 0$, we first need the following result.

\begin{lemma} \label{lem:singMat}
Let $A$ and $B$ be $n \times n$ symmetric matrices with $A \succeq 0$, $B \preceq 0$, and $\trace(A B) = 0$. 
If $B$ has at least one nonzero eigenvalue, then, $A$ is singular.
\end{lemma}

\begin{proof}
Since $A \succeq 0$ and $B \preceq 0$, the product $A B$ has only non-positive eigenvalues $\lambda_{i} \leq 0$, where, 
for $i = 1, 2, \ldots, n$, $\lambda_i \in \spec(A B)$~\cite{horn2012matrix}. 
It follows from $\trace(A B) = 0 = \sum_{i=1}^{n} \lambda_i$ that $\spec(A B) = \{0\}$. 
Now we assume $A \succ 0$. 
Since the matrix $AB$ is similar to $A^{\frac{1}{2}} B A^{\frac{1}{2}}$, we have 
$\spec(A B) = \spec \big( A^{\frac{1}{2}} B A^{\frac{1}{2}} \big) = \{0\}$. 
Furthermore, $B$ is congruent to $A^{\frac{1}{2}} B A^{\frac{1}{2}}$, and thus both matrices have the same number of zero and nonzero eigenvalues. 
It follows that all eigenvalues of $B$ are zero, which is a contradiction. 
Therefore, $A$ has to be singular. \hfill
\end{proof}

\begin{theorem} \label{thm:lossless}
Let Assumption~\ref{asm:A-inv} hold, and let the Slater's condition~\cite{boyd2004convex} be satisfied. 
Then, the optimal solution to the SDP problem \eqref{opt:convex} satisfies, for all $k = 0, 1, \ldots, N-1$, $C_k = 0$. 
\end{theorem}

\begin{proof}
Using the Lagrange multipliers $M_k$ and $\Lambda_k$ for the constraints in $C_k$ and $G_k$, respectively, we define a Lagrangian function as 
\begin{multline*}
\mathcal{L}(\Sigma_{k}, U_k, Y_k, M_k, \Lambda_k) = \bar{J}_{\Sigma} + \sum_{k = 0}^{N-1} \trace \big( M_k\t C_k \big) 
\\* 
+ \trace \big( \Lambda_k\t G_k \big). 
\end{multline*}
Since Slater's condition holds, it follows that strong duality and the KKT conditions also hold~\cite{boyd2004convex}.
The first-order optimality conditions are 
\begin{subequations}\label{opt:opt_cond}
\begin{align}
& \frac{\partial \mathcal{L}}{\partial \Sigma_{k}} = Q_k - \Sigma_k^{-1} U_k\t M_k U_k \Sigma_k^{-1} + A_k\t \Lambda_k A_k - \Lambda_{k-1} = 0, \label{opt:opt_cond_Sigma} \\
& \frac{\partial \mathcal{L}}{\partial U_k} = 2 M_k U_k \Sigma_k^{-1} + 2 B_k\t \Lambda_k A_k = 0, \label{opt:opt_cond_U} \\ 
&\frac{\partial \mathcal{L}}{\partial Y_k} = R_k - M_k + B_k\t \Lambda_k B_k = 0, \label{opt:opt_cond_Y} \\ 
& G_k = 0, \quad C_k \preceq 0, \quad M_k \succeq 0, \label{opt:opt_cond_G_C_Mu} \\
& \trace \big(M_k\t C_k \big) = 0. \label{opt:opt_cond_comp_slack}
\end{align}
\end{subequations}
%
Next, we prove that the optimal solution to \eqref{opt:convex} satisfies, for all $k$, $C_k = 0$. 
To this end, assume that, for some $k$, $C_k$ has at least one nonzero eigenvalue. 
In light of Lemma~\ref{lem:singMat} and the complementary slackness condition \eqref{opt:opt_cond_comp_slack}, it follows that $M_k$ has to be singular. 
The optimality condition \eqref{opt:opt_cond_U} can be rewritten as 
$ 
B_k\t \Lambda_k = - M_k U_k \Sigma_k^{-1} A_k^{-1}
$. 
Substituting the above equation into \eqref{opt:opt_cond_Y} yields 
\begin{equation} \label{opt:contradicion}
R_k = M_k \big( I_{p} + U_k \Sigma_k^{-1} A_k^{-1} B_k \big). 
\end{equation}
Calculating the determinants of both sides of \eqref{opt:contradicion}, yields 
\[ 
\det(R_k) = \det(M_k) \det\big( I_{p} + U_k \Sigma_k^{-1} A_k^{-1} B_k \big) = 0. 
\]
This clearly contradicts the fact that $R_k \succ 0$. 
Therefore, at the optimal solution to problem \eqref{opt:convex}, the matrix $C_k$ has all its eigenvalues equal to zero. 
Since $C_k \preceq 0$, it follows that, for all $k = 0,1,\ldots, N-1$, $C_k = 0$. \hfill
\end{proof}

\begin{corollary} \label{cor:lossless}
Let Assumptions~\ref{asm:A-inv}-\ref{asm:G-D} hold.
Then, the optimal solution to the SDP problem \eqref{opt:convex} satisfies, for all $k = 0, 1, \ldots, N-1$, $C_k = 0$. 
\end{corollary}

\begin{proof}
For Slater's condition, it suffices to find some strictly feasible values of $\Sigma_{k}$, $U_k$, and $Y_k$, such that, for all $k = 0, 1, \ldots, N-1$, $C_k \prec 0$ and $G_k = 0$. 
%
For some $\varepsilon > 0$, let $C_k = - \varepsilon I_p$. 
This choice corresponds to a new system, where the noise coefficient matrices $D_k$ are replaced with $\bar{D}_k$ such that, for some $\varepsilon > 0$, 
$\bar{D}_k \bar{D}_k\t = D_k D_k\t + \varepsilon B_k B_k\t$. 
Then, it follows from \cite[Theorem~3]{liu2024reach} that, for a sufficiently small $\varepsilon$, all terminal state covariances $\bar{\Sigma}_{N} \succ D_{N-1} D_{N-1}\t$ are reachable when the new system starts from a given initial state covariance $\bar{\Sigma}_{0} \succ 0$. 
Thus, the Slater's condition holds. 
\hfill
\end{proof}


\section{Numerical Example} \label{sec:example}

\begin{figure*}[!h]
\centering
\begin{subfigure}[h]{0.48\textwidth}
\centering
\includegraphics[width=\textwidth]{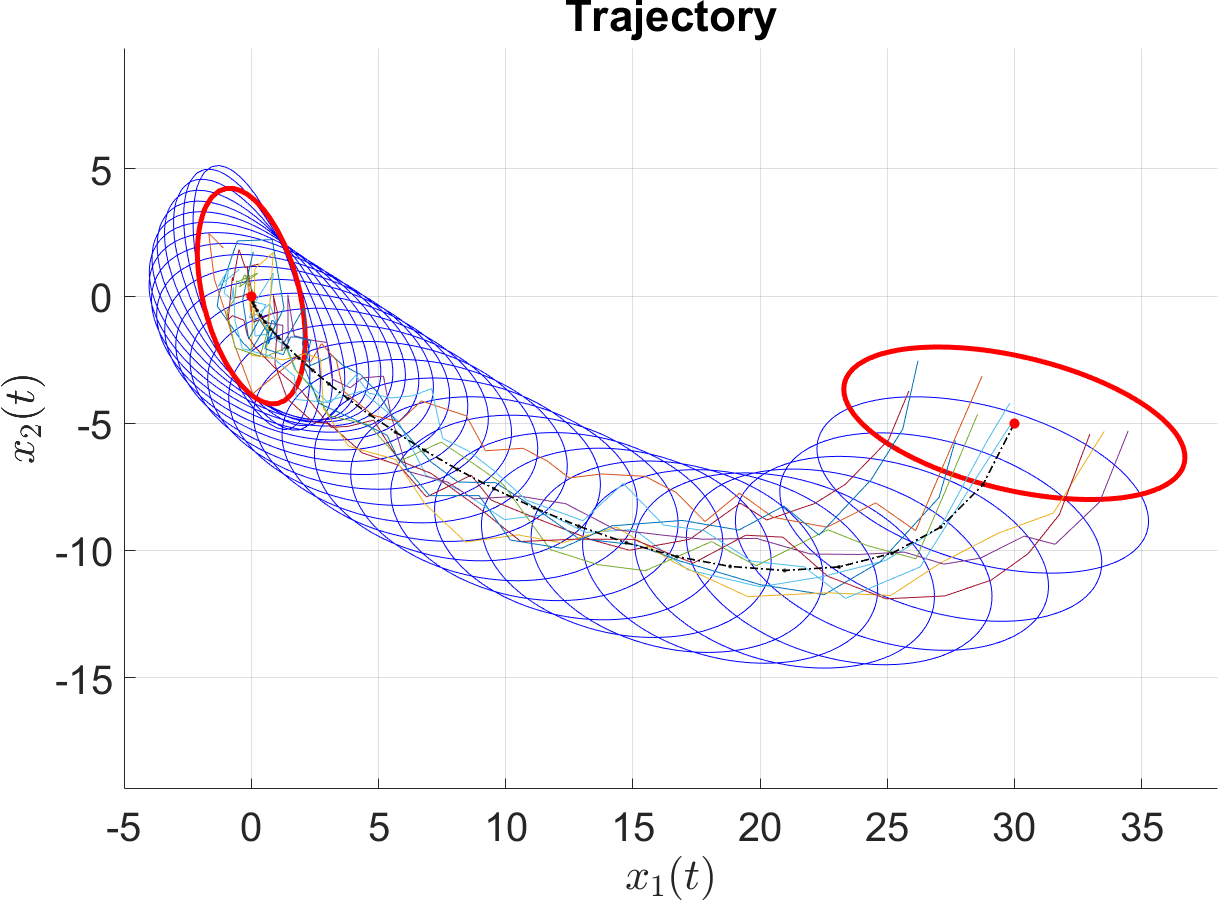}
\caption{Newton's method.}
\label{fig:Newton}
\end{subfigure}
\hfill
\begin{subfigure}[h]{0.48\textwidth}
\centering
\includegraphics[width=\textwidth]{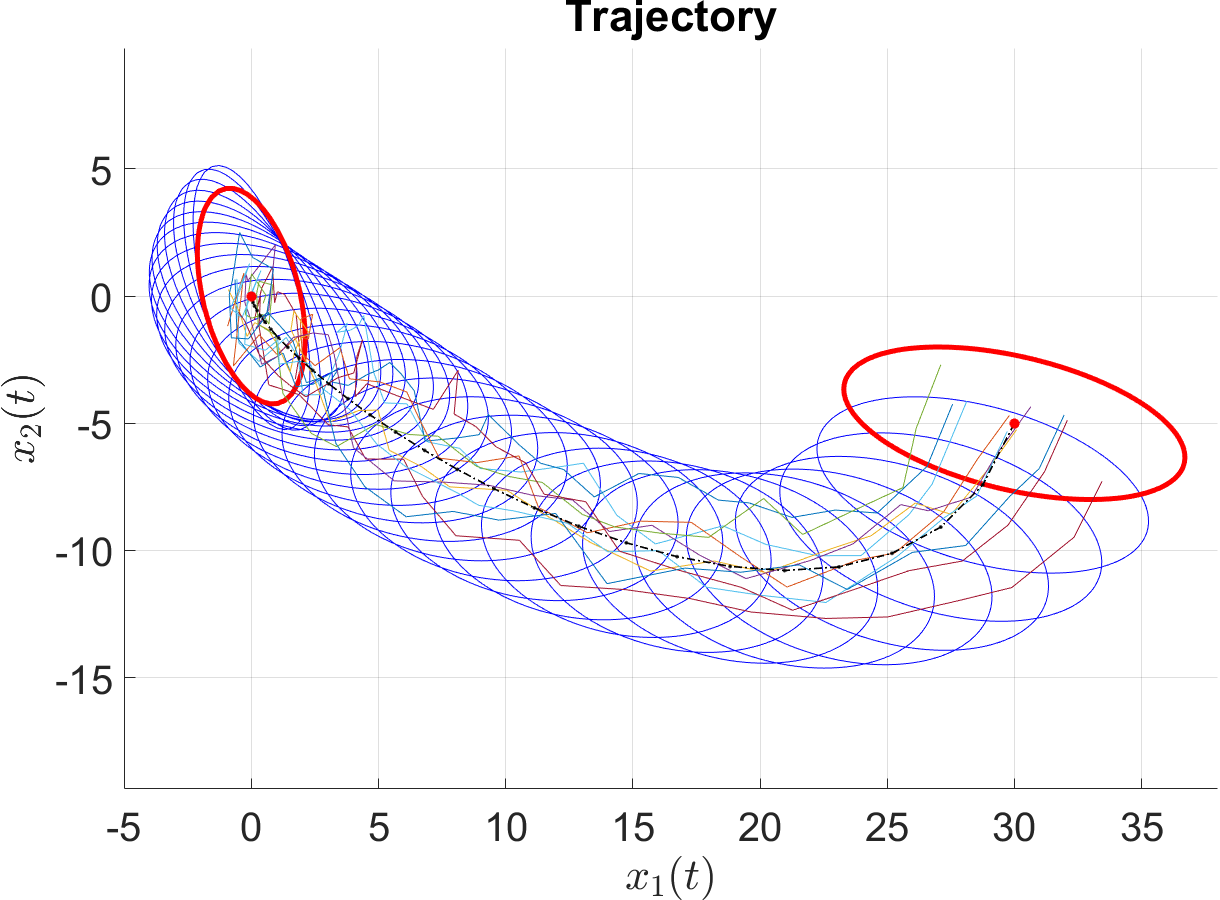}
\caption{SDP approach.}
\label{fig:SDP}
\end{subfigure}
\caption{Ten sample paths with the optimal state mean and three-standard-deviation tolerance region.}
\label{fig:2D}
\end{figure*}

In this section, we use a two-dimensional system to illustrate our theoretical results. 
Specifically, we consider a linear time-invariant system over a horizon of $N = 30$ steps, and we seek to design a feedback controller to steer all trajectories with initial conditions distributed according to a given Gaussian distribution to another given final Gaussian distribution within the time horizon. 
It is assumed that the state-space matrices are given by 
\[ 
A = \begin{bmatrix} 1 & 0.2 \\ 0 & 1 \end{bmatrix}, \qquad 
B = \begin{bmatrix} 0.02 \\ 0.2 \end{bmatrix}, \qquad 
D = \begin{bmatrix} 0.4 & 0 \\ 0.4 & 0.6 \end{bmatrix}. 
\]
The boundary conditions and cost function parameters are 
\begin{align*}
\bar{\Sigma}_0 &= \begin{bmatrix} 5 & -1 \\ -1 & 1 \end{bmatrix}, \hspace{2mm} 
\bar{\Sigma}_N = \begin{bmatrix} 0.5 & -0.4 \\ -0.4 & 2 \end{bmatrix}, \hspace{2mm} 
\bar{\mu}_0 = \begin{bmatrix} 30 \\ -5 \end{bmatrix}, 
\\ 
\bar{\mu}_N &= \begin{bmatrix} 0 \\ 0 \end{bmatrix}, \hspace{2mm} 
Q_k \equiv 0.5 I_2, \hspace{2mm} 
R_k \equiv 1. 
\end{align*}

We compute the optimal controller using both approaches in Section~\ref{sec:compute}. 
Namely, we first compute the optimal control using Newton's method, as described in Section~\ref{subsec:Newton}. 
Ten sample paths along with their mean and their covariance ellipses are shown in Figure~\ref{fig:Newton}, 
where the initial and target state means (respectively, the initial and target three-standard-deviation tolerance regions) are marked with red dots (respectively, red ellipses), 
and the optimal state mean (respectively, the optimal three-standard-deviation tolerance region) at each time step is marked with a black dash (respectively, a blue ellipse). 
Then, using the SDP method described in Section~\ref{subsec:SDP}, we also solve the corresponding optimization problem using YALMIP~\cite{lofberg2004yalmip} and MOSEK~\cite{aps2019mosek}.
The optimal solution is shown in Figure~\ref{fig:SDP}, where ten sample paths of the optimal process are plotted, along with the optimal state mean and three-standard-deviation ellipses. 
As expected, the two solutions are identical. 
However, our experience has shown that Newton's method is very sensitive to initialization and does not scale well, with higher dimensional systems and longer time horizons. 
On the contrary, the SDP method is numerically more robust and scales much better with the problem parameters.

\section{Concluding Remarks}

We have established the existence and uniqueness of the optimal control for the covariance steering of a discrete-time linear stochastic system. 
We have also shown the separation of the optimal mean and the optimal covariance steering problems, 
and have demonstrated that the exact covariance steering problem can be recast as a convex semi-definite programming problem, which can be solved efficiently using standard convex solvers. 
In the process, we have investigated various properties of a matrix Riccati difference equation that shows up in the solution for the optimal control. 
In the future, we would like to study the optimal covariance steering problem for stochastic systems subject to multiplicative noise and chance constraints, and to develop efficient data-driven algorithms for estimating the noise covariance for real-life engineering problems.


\bibliographystyle{IEEEtran}
\bibliography{final}

\begin{thebibliography}{10}
\providecommand{\url}[1]{#1}
\csname url@samestyle\endcsname
\providecommand{\newblock}{\relax}
\providecommand{\bibinfo}[2]{#2}
\providecommand{\BIBentrySTDinterwordspacing}{\spaceskip=0pt\relax}
\providecommand{\BIBentryALTinterwordstretchfactor}{4}
\providecommand{\BIBentryALTinterwordspacing}{\spaceskip=\fontdimen2\font plus
\BIBentryALTinterwordstretchfactor\fontdimen3\font minus
  \fontdimen4\font\relax}
\providecommand{\BIBforeignlanguage}[2]{{%
\expandafter\ifx\csname l@#1\endcsname\relax
\typeout{** WARNING: IEEEtran.bst: No hyphenation pattern has been}%
\typeout{** loaded for the language `#1'. Using the pattern for}%
\typeout{** the default language instead.}%
\else
\language=\csname l@#1\endcsname
\fi
#2}}
\providecommand{\BIBdecl}{\relax}
\BIBdecl

\bibitem{hotz1985covariance}
A.~F. Hotz and R.~E. Skelton, ``A covariance control theory,'' in \emph{Proc.
  {IEEE} Conf. Decision Control}, Lauderdale, FL, 1985, pp. 552--557.

\bibitem{collins1985covariance}
E.~Collins and R.~Skelton, ``Covariance control of discrete systems,'' in
  \emph{Proc. {IEEE} Conf. Decision Control}, Lauderdale, FL, 1985, pp.
  542--547.

\bibitem{collins1987theory}
E.~Collins and R.~Skelton, ``A theory of state covariance assignment for
  discrete systems,'' \emph{{IEEE} Trans. Autom. Control}, vol.~32, no.~1, pp.
  35--41, 1987.

\bibitem{hsieh1990all}
C.~Hsieh and R.~E. Skelton, ``All covariance controllers for linear
  discrete-time systems,'' \emph{{IEEE} Trans. Autom. Control}, vol.~35, no.~8,
  pp. 908--915, 1990.

\bibitem{xu1992improved}
J.-H. Xu and R.~E. Skelton, ``An improved covariance assignment theory for
  discrete systems,'' \emph{{IEEE} Trans. Autom. Control}, vol.~37, no.~10, pp.
  1588--1591, 1992.

\bibitem{zhang2016parametric}
Q.~Zhang, Z.~Wang, and H.~Wang, ``Parametric covariance assignment using a
  reduced-order closed-form covariance model,'' \emph{Syst. Sci. Control.
  Eng.}, vol.~4, no.~1, pp. 78--86, 2016.

\bibitem{grigoriadis1997minimum}
K.~M. Grigoriadis and R.~E. Skelton, ``Minimum-energy covariance controllers,''
  \emph{Automatica}, vol.~33, no.~4, pp. 569--578, 1997.

\bibitem{bakolas2018finite}
E.~Bakolas, ``Finite-horizon covariance control for discrete-time stochastic
  linear systems subject to input constraints,'' \emph{Automatica}, vol.~91,
  pp. 61--68, 2018.

\bibitem{okamoto2019input}
K.~Okamoto and P.~Tsiotras, ``Input hard constrained optimal covariance
  steering,'' in \emph{Proc. {IEEE} Conf. Decision Control}, Nice, France,
  2019, pp. 3497--3502.

\bibitem{kotsalis2021convex}
G.~Kotsalis, G.~Lan, and A.~S. Nemirovski, ``Convex optimization for
  finite-horizon robust covariance control of linear stochastic systems,''
  \emph{SIAM J. Control Optim.}, vol.~59, no.~1, pp. 296--319, 2021.

\bibitem{pilipovsky2021covariance}
J.~Pilipovsky and P.~Tsiotras, ``Covariance steering with optimal risk
  allocation,'' \emph{IEEE Trans. Aerosp. Electron. Syst.}, vol.~57, no.~6, pp.
  3719--3733, 2021.

\bibitem{balci2023covariance}
I.~M. Balci and E.~Bakolas, ``Covariance steering of discrete-time linear
  systems with mixed multiplicative and additive noise,'' in \emph{Proc. Amer.
  Control Conf.}, San Diego, CA, 2023, pp. 2586--2591.

\bibitem{balci2022exact}
I.~M. Balci and E.~Bakolas, ``Exact {SDP} formulation for discrete-time
  covariance steering with {W}asserstein terminal cost,'' 2022,
  arXiv:2205.10740.

\bibitem{okamoto2019stochastic}
K.~Okamoto and P.~Tsiotras, ``Stochastic model predictive control for
  constrained linear systems using optimal covariance steering,'' 2019,
  arXiv:1905.13296.

\bibitem{yin2022trajectory}
J.~Yin, Z.~Zhang, E.~Theodorou, and P.~Tsiotras, ``Trajectory distribution
  control for model predictive path integral control using covariance
  steering,'' in \emph{Proc. {IEEE} Int. Conf. Robot. Autom.}, Philadelphia,
  PA, 2022, pp. 1478--1484.

\bibitem{sivaramakrishnan2022distribution}
V.~Sivaramakrishnan, J.~Pilipovsky, M.~M. Oishi, and P.~Tsiotras,
  ``Distribution steering for discrete-time linear systems with general
  disturbances using characteristic functions,'' in \emph{Proc. Amer. Control
  Conf.}, Atlanta, GA, 2022, pp. 4183--4190.

\bibitem{chen2016I}
Y.~Chen, T.~T. Georgiou, and M.~Pavon, ``Optimal steering of a linear
  stochastic system to a final probability distribution, part {I},''
  \emph{{IEEE} Trans. Autom. Control}, vol.~61, no.~5, pp. 1158--1169, 2016.

\bibitem{chen2016II}
Y.~Chen, T.~T. Georgiou, and M.~Pavon, ``Optimal steering of a linear
  stochastic system to a final probability distribution, part {II},''
  \emph{{IEEE} Trans. Autom. Control}, vol.~61, no.~5, pp. 1170--1180, 2016.

\bibitem{chen2018III}
Y.~Chen, T.~T. Georgiou, and M.~Pavon, ``Optimal steering of a linear
  stochastic system to a final probability distribution, part {III},''
  \emph{{IEEE} Trans. Autom. Control}, vol.~63, no.~9, pp. 3112--3118, 2018.

\bibitem{liu2024add}
F.~Liu and P.~Tsiotras, ``Optimal covariance steering for continuous-time
  linear stochastic systems with martingale additive noise,'' \emph{{IEEE}
  Trans. Autom. Control}, vol.~69, no.~4, pp. 2591--2597, 2024.

\bibitem{liu2024mult}
F.~Liu and P.~Tsiotras, ``Optimal covariance steering for continuous-time
  linear stochastic systems with multiplicative noise,'' \emph{{IEEE} Trans.
  Autom. Control}, 2024, in press.

\bibitem{ridderhof2018uncertainty}
J.~Ridderhof and P.~Tsiotras, ``Uncertainty quantication and control during
  {M}ars powered descent and landing using covariance steering,'' in \emph{AIAA
  Guidance, Navigation, Control Conf.}, Kissimmee, FL, 2018.

\bibitem{ridderhof2020fuel}
J.~Ridderhof, J.~Pilipovsky, and P.~Tsiotras, ``Chance-constrained covariance
  control for low-thrust minimum-fuel trajectory optimization,'' in
  \emph{AAS/AIAA Astrodynamics Specialist Conf.}, South Lake Tahoe, CA, 2020.

\bibitem{okamoto2019optimal}
K.~Okamoto and P.~Tsiotras, ``Optimal stochastic vehicle path planning using
  covariance steering,'' \emph{IEEE Robot. Autom. Lett.}, vol.~4, no.~3, pp.
  2276--2281, 2019.

\bibitem{zheng2022belief}
D.~Zheng, J.~Ridderhof, P.~Tsiotras, and A.-a. Agha-mohammadi, ``Belief space
  planning: a covariance steering approach,'' in \emph{Proc. {IEEE} Int. Conf.
  Robot. Autom.}, Philadelphia, PA, 2022, pp. 11\,051--11\,057.

\bibitem{liu2024reach}
F.~Liu and P.~Tsiotras, ``Reachability and controllability analysis of the
  state covariance for linear stochastic systems,'' \emph{Automatica}, 2024,
  under review, arXiv:2406.14740.

\bibitem{horn2012matrix}
R.~A. Horn and C.~R. Johnson, \emph{Matrix Analysis}.\hskip 1em plus 0.5em
  minus 0.4em\relax Cambridge University Press, 2012.

\bibitem{ito2023maximum}
K.~Ito and K.~Kashima, ``Maximum entropy optimal density control of
  discrete-time linear systems and {S}chr{\"o}dinger bridges,'' \emph{{IEEE}
  Trans. Autom. Control}, 2023, in press.

\bibitem{henderson1981deriving}
H.~V. Henderson and S.~R. Searle, ``On deriving the inverse of a sum of
  matrices,'' \emph{SIAM Rev.}, vol.~23, no.~1, pp. 53--60, 1981.

\bibitem{krantz2002implicit}
S.~G. Krantz and H.~R. Parks, \emph{The Implicit Function Theorem: History,
  Theory, and Applications}.\hskip 1em plus 0.5em minus 0.4em\relax Springer,
  2002.

\bibitem{zhang2005schur}
F.~Zhang, \emph{The Schur Complement and Its Applications}.\hskip 1em plus
  0.5em minus 0.4em\relax Springer, 2005.

\bibitem{okamoto2018optimal}
K.~Okamoto, M.~Goldshtein, and P.~Tsiotras, ``Optimal covariance control for
  stochastic systems under chance constraints,'' \emph{IEEE Control Syst.
  Lett.}, vol.~2, no.~2, pp. 266--271, 2018.

\bibitem{boyd2004convex}
S.~Boyd and L.~Vandenberghe, \emph{Convex Optimization}.\hskip 1em plus 0.5em
  minus 0.4em\relax Cambridge University Press, 2004.

\bibitem{lofberg2004yalmip}
J.~Lofberg, ``{YALMIP}: a toolbox for modeling and optimization in {MATLAB},''
  in \emph{IEEE Int. Conf. Robot. Autom.}, New Orleans, LA, 2004, pp. 284--289.

\bibitem{aps2019mosek}
M.~ApS, ``{MOSEK} optimization toolbox for {MATLAB},'' \emph{User’s Guide and
  Reference Manual, Version 4}, 2019.

\bibitem{freiling1996generalized}
G.~Freiling, G.~Jank, and H.~Abou-Kandil, ``Generalized {R}iccati difference
  and differential equations,'' \emph{Linear Algebra Its Appl.}, vol. 241, pp.
  291--303, 1996.

\end{thebibliography}


\section*{Appendix~A}

In this appendix, we show two auxiliary results for proving Proposition~\ref{prp:suff-Q}, Proposition~\ref{prp:pi-monotone}, and Lemma~\ref{lem:phi-11-12-inv-mono}.


\begin{lemma} \label{lem:A-bar-B-contr}
Suppose Assumption~\ref{asm:A-inv}, Assumption~\ref{asm:G-inv}, and Property~\ref{pty:pi-pos} hold. 
Then, 
\begin{equation*}
\sum_{i=0}^{N-1} \Phi_{\bar{A}}(N, i+1) B_{i} \big(R_{i} + B_{i}\t \Pi_{i+1} B_{i} \big)^{-1} B_{i}\t \Phi_{\bar{A}}(N, i+1)\t \succ 0. 
\end{equation*}
\end{lemma}

\begin{proof} 
It is clear that we can write 
$
\sum_{i=0}^{N-1} \Phi_{\bar{A}}(N, i+1) B_{i} \big(R_{i} + B_{i}\t \Pi_{i+1} B_{i} \big)^{-1} B_{i}\t \Phi_{\bar{A}}(N, i+1)\t 
= 
\Xi \Lambda \Xi\t
$, 
where 
$\Xi \triangleq 
\big[
B_{N-1} ~~ \bar{A}_{N-1} B_{N-2} ~~ \cdots ~~ \bar{A}_{N-1} \bar{A}_{N-2} \cdots \bar{A}_{1} B_{0}
\big]$ 
and 
$\Lambda \triangleq 
\blkdiag \big[ (R_{N-1} + B_{N-1}\t \Pi_{N} B_{N-1})^{-1}, (R_{N-2} + B_{N-2}\t \Pi_{N-1} B_{N-2})^{-1}, \ldots, (R_{0} + B_{0}\t \Pi_{1} B_{0})^{-1} \big]$. 
From Property~\ref{pty:pi-pos}, $\Lambda \succ 0$. 
Thus, it suffices to show that $\rank \Xi = n$. 
From Assumption~\ref{asm:G-inv}, the reachability Gramian $G(N, 0) \succ 0$. 
From the Woodbury formula~\cite{horn2012matrix}, we can write $\bar{A}_{k} = A_{k} + B_{k} F_{k}$ for some $F_{k} \in \mathbb{R}^{p \times n}$. 
It follows that 
$\range \Xi 
= \range \big[B_{N-1} ~~ A_{N-1} B_{N-2} ~~ \cdots ~~ A_{N-1} A_{N-2} \cdots A_{1} B_{0}\big] 
= \mathbb{R}^{n}$. 
Therefore, $\rank \Xi = n$. \hfill
\end{proof}


The following result is used to show the monotonicity of the solution to \eqref{eqn:pi-Q} and of the matrix $W_{k,0}$ defined by \eqref{def:W}.

\begin{lemma} \label{lem:monotone}
Let $X \succ 0$ be a $j \times j$ matrix, written as
\begin{equation*}
X \triangleq 
\begin{bmatrix}
X_{1} & X_{2} \\
X_{2}\t & X_{4}
\end{bmatrix}
\succ 0,
\end{equation*}
where $X_{1}$ is an $i \times i$ matrix for some $i < j$. 
Let $Y$ be an $i \times i$ symmetric matrix. 
Then, 
\begin{align*}
X^{-1} &= 
\begin{bmatrix}
X_{1} & X_{2} \\
X_{2}\t & X_{4}
\end{bmatrix}^{-1}
\succ ~ (\text{respectively}, \, \succeq)
\begin{bmatrix}
Y & 0 \\
0 & 0
\end{bmatrix}
\\* 
&\hspace{5mm} 
\iff \quad 
X_{1}^{-1} \succ ~ (\text{respectively}, \, \succeq) ~ Y.
\end{align*}
\end{lemma}

\begin{proof}
The fact of $X \succ 0$ implies that $X_{1} \succ 0$ and its Schur complement $X/X_{1} \triangleq X_{4} - X_{2}\t X_{1}^{-1} X_{2} \succ 0$~\cite{horn2012matrix}. 
Using the expression for $X^{-1}$, we get 
\begin{align*}
&Z \triangleq 
X^{-1} - 
\begin{bmatrix}
Y & 0 \\
0 & 0
\end{bmatrix}
= 
\\* 
&\hspace{-1mm}
\begin{bmatrix} 
\hspace{-0.5mm}
X_{1}^{-1} \hspace{-0.8mm} + \hspace{-0.8mm} X_{1}^{-1} X_{2} ( \hspace{-0.5mm} X/X_{1} \hspace{-0.5mm} )^{-1} X_{2}\t X_{1}^{-1} \hspace{-0.8mm} - \hspace{-0.8mm} Y & - \hspace{-0.8mm} X_{1}^{-1} X_{2} (\hspace{-0.5mm} X/X_{1} \hspace{-0.5mm} )^{-1} \\
- (X/X_{1})^{-1} X_{2}\t X_{1}^{-1} & (X/X_{1})^{-1}
\end{bmatrix} \hspace{-1.5mm} .
\end{align*}
We can check that the Schur complement of $(X/X_{1})^{-1}$ in $Z$ is $X_{1}^{-1} - Y$. 
Since $(X/X_{1})^{-1} \succ 0$, it follows that $Z \succ 0$ (respectively, $Z \succeq 0$) if and only if the Schur complement $X_{1}^{-1} - Y \succ 0$ (respectively, $X_{1}^{-1} - Y \succeq 0$). \hfill
\end{proof}


\section*{Appendix~B}

In this appendix, we give the proof of Proposition~\ref{prp:pi-Q-exist} in 
Section~\ref{subsec:riccati}, along with some related results that are used to prove Theorem~\ref{thm:pi-pos-exist-cond}. 

\begin{proof}[Proof of Proposition~\ref{prp:pi-Q-exist}] 
First, we show that if, for a given initial condition $\Pi_{0}$, equation \eqref{eqn:pi-Q} has a solution from time $0$ to $N$, then, this solution is unique. 
Let $\mathcal{S}_{n}$ be the set of $n \times n$ symmetric matrices. 
In view of \eqref{eqn:pi-Q}, it suffices to show that the map $g: \{\Pi \, | \, \Pi \in \mathcal{S}_{n}, ~ I_{n} + B B\t \Pi ~\text{is invertible}\} \to \mathcal{S}_{n}$ with $B \in \mathbb{R}^{n \times p}$, defined by
\begin{equation*}
g: \Pi \mapsto \Pi \big(I_{n} + B B\t \Pi\big)^{-1}, 
\end{equation*}
is injective. 
Suppose $H = g(\Pi)$ for some $\Pi$. 
If $H$ is invertible, then $\Pi$ must be invertible. 
It follows that $H = \big(\Pi^{-1} + B B\t\big)^{-1}$. 
Then, $\Pi = \big(H^{-1} - B B\t\big)^{-1}$ is unique.

If $H$ is singular, then any $\Pi$ that satisfies $H = g(\Pi)$ must be singular. 
Moreover, for any such $\Pi$, we have $\ker H = \ker \Pi$, since, for some $z \in \mathbb{R}^{n}$, $z\t H = 0$ if and only if $z\t \Pi = 0$. 
Hence, we can quotient out the common kernel of $H$ and $\Pi$. 
Since $H, \Pi \in \mathcal{S}_{n}$, there exists an $n \times n$ orthogonal matrix $V \triangleq [\bar{V} \enspace \tilde{V}]$ such that
\begin{equation*}
H = 
\begin{bmatrix}
\bar{V} & \tilde{V}
\end{bmatrix}
\hspace{-1mm} 
\begin{bmatrix}
\bar{\Lambda} & 0 \\
0 & 0 
\end{bmatrix}
\hspace{-1mm} 
\begin{bmatrix}
\bar{V}\t \\ 
\tilde{V}\t
\end{bmatrix},
~ 
\Pi = 
\begin{bmatrix}
\bar{V} & \tilde{V}
\end{bmatrix}
\hspace{-1mm} 
\begin{bmatrix}
\bar{\Pi} & 0 \\
0 & 0 
\end{bmatrix}
\hspace{-1mm} 
\begin{bmatrix}
\bar{V}\t \\ 
\tilde{V}\t
\end{bmatrix},
\end{equation*}
where $\ker H = \ker \Pi = \range \tilde{V}$ and $\bar{\Lambda}$ is an $r \times r$ diagonal matrix. 
It can be checked that $\bar{\Lambda} = \bar{\Pi} \big(I_{r} + T \bar{\Pi}\big)^{-1}$, where $T$ is the $r$th leading principal submatrix (that is, the upper-left $r \times r$ submatrix) of $V\t B B\t V$. 
Since $\bar{\Pi}$ is invertible by construction, we have $\bar{\Lambda} = \big(\bar{\Pi}^{-1} + T\big)^{-1}$. 
Thus, $\bar{\Pi}$ is unique. 
It follows immediately that $\Pi$ is also unique.

Next, we show the necessary and sufficient condition for the existence of the solution \eqref{sol:pi-Q-0}. 
(Sufficiency) Suppose, for all $k = 0, 1, \ldots, N$, the matrix $\Phi_{11}(k, 0) + \Phi_{12}(k, 0) \Pi_{0}$ is invertible. 
We will show that the matrix $\Pi_{k}$ given by \eqref{sol:pi-Q-0} is a solution of \eqref{eqn:pi-Q}. 
Since $\Phi_{M}(k+1, 0) = M_{k} \Phi_{M}(k, 0)$, it follows that
\begin{multline} \label{eqn:phi-k-k+1}
\begin{bmatrix}
\Phi_{11}(k, 0) + \Phi_{12}(k, 0) \Pi_{0} \\
\Phi_{21}(k, 0) + \Phi_{22}(k, 0) \Pi_{0}
\end{bmatrix}
= 
\\* 
M_{k}^{-1} 
\begin{bmatrix} 
\Phi_{11}(k+1, 0) + \Phi_{12}(k+1, 0) \Pi_{0} \\
\Phi_{21}(k+1, 0) + \Phi_{22}(k+1, 0) \Pi_{0}
\end{bmatrix}.
\end{multline}
Moreover, we can compute that 
\begin{equation} \label{eqn:M-inv}
M_{k}^{-1} = 
\begin{bmatrix}
A_{k}^{-1} & A_{k}^{-1} B_{k} B_{k}\t \\
Q_{k} A_{k}^{-1} & A_{k}\t + Q_{k} A_{k}^{-1} B_{k} B_{k}\t
\end{bmatrix}.
\end{equation}
Putting \eqref{eqn:phi-k-k+1} and \eqref{eqn:M-inv} together, it can be checked that the $\Pi_{k}$ given by \eqref{sol:pi-Q-0} satisfies \eqref{eqn:pi-Q}.

(Necessity) Suppose \eqref{eqn:pi-Q} has a solution from time $0$ to $N$, denoted by $\Pi_{k}$. 
Then, this solution is unique. 
Clearly, $\Phi_{11}(0, 0) + \Phi_{12}(0, 0) \Pi_{0} = I_{n}$ is invertible. 
Assume that there exists $0 \leq k \leq N-1$ such that, for all $i = 0, 1, \ldots, k$, the matrix $\Phi_{11}(i, 0) + \Phi_{12}(i, 0) \Pi_{0}$ is invertible, but the matrix $\Phi_{11}(k+1, 0) + \Phi_{12}(k+1, 0) \Pi_{0}$ is singular. 
In view of the proof for sufficiency, the unique solution of \eqref{eqn:pi-Q} up to time $k$ is given by \eqref{sol:pi-Q-0}. 
From \eqref{eqn:pi-Q}, and since $A_{k}$ is invertible and $\Pi_{k}$ is symmetric, it follows that
\begin{align*}
\Pi_{k} 
&= \Big(A_{k}^{- \mbox{\tiny\sf T}} + \Pi_{k+1} B_{k} B_{k}\t A_{k}^{- \mbox{\tiny\sf T}}\Big)^{-1} 
\\* 
&\hspace{15mm} 
\times \Big(\Pi_{k+1} A_{k} + \big(A_{k}^{- \mbox{\tiny\sf T}} + \Pi_{k+1} B_{k} B_{k}\t A_{k}^{-\mbox{\tiny\sf T}}\big) Q_{k}\Big).
\end{align*}
From \eqref{sol:pi-Q-0} and \eqref{eqn:phi-k-k+1}, it follows that
\begin{align*}
&\Pi_{k} 
= \Big(Q_{k} A_{k}^{-1} \Big[\Phi_{11}(k+1, 0) + \Phi_{12}(k+1, 0) \Pi_{0}\Big] \\*
&\hspace{4mm} + \big(A_{k}\t \hspace{-0.5mm} + \hspace{-0.5mm} Q_{k} A_{k}^{-1} B_{k} B_{k}\t\big) \hspace{-0.6mm} \Big[\Phi_{21}(k \hspace{-0.5mm} + \hspace{-0.5mm} 1, 0) \hspace{-0.5mm} + \hspace{-0.5mm} \Phi_{22}(k \hspace{-0.5mm} + \hspace{-0.5mm} 1, 0) \Pi_{0}\Big]\Big) \\*
&\hspace{4mm} \times \Big(A_{k}^{-1} \Big[\Phi_{11}(k+1, 0) + \Phi_{12}(k+1, 0) \Pi_{0}\Big] \\*
&\hspace{4mm} + A_{k}^{-1} B_{k} B_{k}\t \Big[\Phi_{21}(k+1, 0) + \Phi_{22}(k+1, 0) \Pi_{0}\Big]\Big)^{-1}.
\end{align*}
Equating the above two equations and canceling out the common terms yields 
\begin{multline*}
\Phi_{21}(k+1, 0) + \Phi_{22}(k+1, 0) \Pi_{0} = 
\\* 
\Pi_{k+1} \Big(\Phi_{11}(k+1, 0) + \Phi_{12}(k+1, 0) \Pi_{0}\Big).
\end{multline*}
There exists $\Pi_{k+1}$ such that the above equation holds if and only if 
\begin{multline*}
0 \neq \ker \Big(\Phi_{11}(k+1, 0) + \Phi_{12}(k+1, 0) \Pi_{0}\Big) \subseteq 
\\* 
\ker \Big(\Phi_{21}(k+1, 0) + \Phi_{22}(k+1, 0) \Pi_{0}\Big).
\end{multline*}
It follows that there exists $\xi \in \mathbb{R}^{n}$ with $\xi \neq 0$ such that the right-hand side of \eqref{eqn:phi-k-k+1} times $\xi$ equals zero. 
However, the left-hand side of \eqref{eqn:phi-k-k+1} times $\xi$ is nonzero, since $\Phi_{11}(k, 0) + \Phi_{12}(k, 0) \Pi_{0}$ is invertible. 
We have thus reached a contradiction.

Lastly, we will prove \eqref{sol:pi-Q-s} by induction. 
For $\ell = 0$, we know that \eqref{sol:pi-Q-s} holds because of \eqref{sol:pi-Q-0}. 
Assume that, for $\ell = i$, where $i \leq N-1$, \eqref{sol:pi-Q-s} holds. 
It can be checked that $\Phi_{M}(k, i) = \Phi_{M}(k, i+1) M_{i}$ regardless of $i < k$ or $i \geq k$. 
Thus, for $\ell = i+1$, \eqref{sol:pi-Q-s} holds. \hfill
\end{proof}


The next result on the monotonicity of the solution to \eqref{eqn:pi-Q} is used to establish the invertibility of $\Phi_{11}$ in Lemma~\ref{lem:phi-eqn}.

\begin{proposition} \label{prp:pi-monotone}
Let $P_{k}^{1}$ and $P_{k}^{2}$ be the respective solutions to the following Riccati difference equations
\begin{align*}
P_{k}^{1} &= A_{k}\t P_{k+1}^{1} \big(I_{n} + \Upsilon_{k} P_{k+1}^{1}\big)^{-1} A_{k} + \Omega_{k}, \quad P_{0}^{1} = 0_{n \times n}, \\
P_{k}^{2} &= A_{k}\t P_{k+1}^{2} \big(I_{n} + \Gamma_{k} P_{k+1}^{2}\big)^{-1} A_{k} + Q_{k}, \quad P_{0}^{2} = 0_{n \times n}, 
\end{align*}
where $A_{k}$ is invertible and $\Upsilon_{k}, \Gamma_{k}, \Omega_{k}, Q_{k} \succeq 0$ over the maximal integer time intervals of existence of the respective solutions $P_{k}^{1}$ and $P_{k}^{2}$. 
Let $\mathcal{I} \subseteq \mathbb{Z}$ denote the common integer time intervals of existence of $P_{k}^{1}$ and $P_{k}^{2}$. 
If, for all $k \in \mathcal{I}$, $\Omega_{k} \succeq Q_{k}$ and $\Upsilon_{k} \preceq \Gamma_{k}$, then, 
\begin{subequations} \label{ineq:pi-monotone}
\begin{align}
& P_{k}^{1} \succeq P_{k}^{2} \succeq 0, \quad k \in \mathcal{I}, ~ k \leq 0, \label{ineq:pi-monotone-1} \\
& P_{k}^{1} \preceq P_{k}^{2} \preceq 0, \quad k \in \mathcal{I}, ~ k \geq 0. \label{ineq:pi-monotone-2}
\end{align}
\end{subequations}
\end{proposition}

\begin{proof}
For the case $k \leq 0$, \eqref{ineq:pi-monotone-1} follows directly from~\cite{freiling1996generalized}. 
We will prove \eqref{ineq:pi-monotone-2} by induction. 
To this end, let $P_{0}^{1} = P_{0}^{2} = 0$; thus, for $k = 0$, \eqref{ineq:pi-monotone-2} holds. 
Assume, for $k = i \geq 0$, that \eqref{ineq:pi-monotone-2} holds. 
We will show that, for $k = i+1 \in \mathcal{I}$, \eqref{ineq:pi-monotone-2} also holds. 
Since $\Omega_{i} \succeq Q_{i}$, we have 
\begin{equation} \label{ineq:p-pi}
P_{i+1}^{1} \big(I_{n} + \Upsilon_{i} P_{i+1}^{1}\big)^{-1} \preceq 
P_{i+1}^{2} \big(I_{n} + \Gamma_{i} P_{i+1}^{2}\big)^{-1} \preceq 0.
\end{equation}
It follows that $\ker P_{i+1}^{1} \subseteq \ker P_{i+1}^{2}$.

If $P_{i+1}^{1}$ and $P_{i+1}^{2}$ are invertible, it follows that 
$\big(\big(P_{i+1}^{1}\big)^{-1} + \Upsilon_{i}\big)^{-1} \preceq 
\big(\big(P_{i+1}^{2}\big)^{-1} + \Gamma_{i}\big)^{-1} \prec 0$. 
Hence, $\big(P_{i+1}^{2}\big)^{-1} + \Gamma_{i} \preceq \big(P_{i+1}^{1}\big)^{-1} + \Upsilon_{i} \prec 0$. 
Since $\Gamma_{k} \succeq \Upsilon_{k} \succeq 0$, it follows that $\big(P_{i+1}^{2}\big)^{-1} \preceq \big(P_{i+1}^{1}\big)^{-1} \prec 0$. 
Thus, we have $P_{i+1}^{1} \preceq P_{i+1}^{2} \prec 0$.

Without loss of generality, assume $P_{i+1}^{1}$ and $P_{i+1}^{2}$ are singular and $\ker P_{i+1}^{1}$ is a proper subset of $\ker P_{i+1}^{2}$. 
The cases when $P_{i+1}^{1}$ is invertible and $P_{i+1}^{2}$ is singular or when $\ker P_{i+1}^{1} = \ker P_{i+1}^{2}$ can be proved in a similar way. 
Since $\ker P_{i+1}^{1} \subseteq \ker P_{i+1}^{2}$, there exists an $n \times n$ orthogonal matrix $U \triangleq [\bar{U} \enspace \tilde{U}]$ such that
\begin{equation*}
P_{i+1}^{1} \hspace{-0.7mm} = \hspace{-0.7mm} 
\begin{bmatrix}
\bar{U} & \tilde{U}
\end{bmatrix}
\hspace{-1.5mm}
\begin{bmatrix}
\bar{P}^{1} & 0 \\
0 & 0 
\end{bmatrix}
\hspace{-2mm}
\begin{bmatrix}
\bar{U}\t \\ 
\tilde{U}\t
\end{bmatrix} 
\hspace{-1mm} ,
P_{i+1}^{2} \hspace{-0.7mm} = \hspace{-0.7mm} 
\begin{bmatrix}
\bar{U} & \tilde{U}
\end{bmatrix}
\hspace{-1.5mm}
\begin{bmatrix}
\bar{P}^{2} & 0 \\
0 & 0 
\end{bmatrix}
\hspace{-2mm}
\begin{bmatrix}
\bar{U}\t \\ 
\tilde{U}\t
\end{bmatrix} \hspace{-1mm} ,
\end{equation*}
where $\ker P_{i+1}^{1} = \range \tilde{U}$ and $\bar{P}^{1}$ is an $r \times r$ diagonal matrix. 
It suffices to show that $\bar{P}^{1} \preceq \bar{P}^{2} \preceq 0$.

Since $\ker P_{i+1}^{1} \neq \ker P_{i+1}^{2}$, there exists an $r \times r$ orthogonal matrix $V \triangleq [\hat{V} \enspace \tilde{V}]$ such that
\begin{equation*}
\bar{P}^{2} = 
\begin{bmatrix}
\hat{V} & \tilde{V}
\end{bmatrix}
\begin{bmatrix}
\hat{P}^{2} & 0 \\
0 & 0 
\end{bmatrix}
\begin{bmatrix}
\hat{V}\t \\ 
\tilde{V}\t
\end{bmatrix},
\end{equation*}
where $\ker \bar{P}^{2} = \range \tilde{V}$ and $\hat{P}^{2}$ is an $\rho \times \rho$ diagonal matrix. 
To this end, it suffices to show that 
\begin{equation*}
V\t \bar{P}^{1} V \preceq 
\begin{bmatrix}
\hat{P}^{2} & 0 \\
0 & 0 
\end{bmatrix} 
\preceq 0. 
\end{equation*}

Multiplying \eqref{ineq:p-pi} by $U\t$ on the left and by $U$ on the right yields the matrix inequality
\begin{equation} \label{ineq:p-pi-bar}
\Big( \hspace{-0.7mm} \big(\bar{P}^{1}\big)^{-1} + \bar{\Upsilon}\Big)^{-1} = 
\bar{P}^{1} \big(I_{r} + \bar{\Upsilon} \bar{P}^{1}\big)^{-1} \preceq 
\bar{P}^{2} \big(I_{r} + \bar{\Gamma} \bar{P}^{2}\big)^{-1} \preceq 0,
\end{equation}
where $\bar{\Upsilon}$ and $\bar{\Gamma}$ are the $r$th leading principal submatrices of $U\t \Upsilon_{i} U$ and $U\t \Gamma_{i} U$, respectively. 
Multiplying \eqref{ineq:p-pi-bar} by $V\t$ on the left and by $V$ on the right yields 
\begin{equation} \label{ineq:p-pi-hat}
\Big(V\t \big(\bar{P}^{1}\big)^{-1} V + V\t \bar{\Upsilon} V\Big)^{-1} \preceq 
\begin{bmatrix}
\Big(\big(\hat{P}^{2}\big)^{-1} + \hat{\Gamma}\Big)^{-1} & 0 \\
0 & 0 
\end{bmatrix} 
\preceq 0,
\end{equation}
where $\hat{\Gamma}$ is the $\rho$th leading principal submatrix of $V\t \bar{\Gamma} V$. 
Let 
\begin{equation*}
V\t \big(\bar{P}^{1}\big)^{-1} V \triangleq 
\begin{bmatrix}
\hat{Z} & \hat{Z}_{2} \\
\hat{Z}_{2}\t & \hat{Z}_{4}
\end{bmatrix}
, \quad
V\t \bar{\Upsilon} V \triangleq 
\begin{bmatrix}
\hat{\Upsilon} & \hat{\Upsilon}_{2} \\
\hat{\Upsilon}_{2}\t & \hat{\Upsilon}_{4}
\end{bmatrix}.
\end{equation*}
Since $\big(V\t \big(\bar{P}^{1}\big)^{-1} V + V\t \bar{\Upsilon} V\big)^{-1} \prec 0$, the matrix inequality \eqref{ineq:p-pi-hat} and Lemma~\ref{lem:monotone} imply that 
$
\big(\hat{Z} + \hat{\Upsilon}\big)^{-1} \preceq \big((\hat{P}^{2})^{-1} + \hat{\Gamma}\big)^{-1} \prec 0
$. 
It follows that $\big(\hat{P}^{2}\big)^{-1} + \hat{\Gamma} \preceq \hat{Z} + \hat{\Upsilon} \prec 0$. 
Since $\hat{\Gamma} \succeq \hat{\Upsilon} \succeq 0$, it follows that $\big(\hat{P}^{2}\big)^{-1} \preceq \hat{Z} \prec 0$. 
Thus, we have $\hat{Z}^{-1} \preceq \hat{P}^{2} \prec 0$. 
The fact that $\big(V\t (\bar{P}^{1})^{-1} V + V\t \bar{\Upsilon} V\big)^{-1} \prec 0$ implies that $V\t \bar{P}^{1} V \prec 0$. 
It then follows from Lemma~\ref{lem:monotone} that 
\begin{equation*}
V\t \bar{P}^{1} V = 
\begin{bmatrix}
\hat{Z} & \hat{Z}_{2} \\
\hat{Z}_{2}\t & \hat{Z}_{4}
\end{bmatrix}^{-1}
\preceq 
\begin{bmatrix}
\hat{P}^{2} & 0 \\
0 & 0 
\end{bmatrix} 
\preceq 0. 
\end{equation*}
This completes the proof. \hfill
\end{proof}


\section*{Appendix~C}

In this appendix, we present some useful properties of the state transition matrix $\Phi_{M}$ defined by \eqref{eqn:phi-M-blk}. 
These properties are used to prove Theorem~\ref{thm:exist-unique-Q}, Theorem~\ref{thm:pi-pos-exist-cond}, and Theorem~\ref{thm:mean-covar}.

Let
\begin{equation*}
\Phi_{M}(k, s)^{-1} = \Phi_{M}(s, k) \triangleq 
\begin{bmatrix}
\Phi_{11}(s, k) & \Phi_{12}(s, k) \\
\Phi_{21}(s, k) & \Phi_{22}(s, k)
\end{bmatrix}.
\end{equation*}


\begin{lemma} \label{lem:phi-eqn}
Assume, for all $k \in \mathbb{Z}$, $A_{k}$ is invertible. 
Then, for all $k, s \in \mathbb{Z}$, 
\begin{subequations} \label{eqn:phi-blk}
\begin{align} 
\Phi_{12}(k, s)\t \Phi_{22}(k, s) &= \Phi_{22}(k, s)\t \Phi_{12}(k, s), \label{eqn:phi-blk-a} \\
\Phi_{21}(k, s)\t \Phi_{11}(k, s) &= \Phi_{11}(k, s)\t \Phi_{21}(k, s), \label{eqn:phi-blk-b} \\
\Phi_{12}(k, s) \Phi_{11}(k, s)\t &= \Phi_{11}(k, s) \Phi_{12}(k, s)\t, \label{eqn:phi-blk-c} \\
\Phi_{21}(k, s) \Phi_{22}(k, s)\t &= \Phi_{22}(k, s) \Phi_{21}(k, s)\t, \label{eqn:phi-blk-d} \\
\Phi_{11}(k, s)\t \Phi_{22}(k, s) &- \Phi_{21}(k, s)\t \Phi_{12}(k, s) = I_{n}, \label{eqn:phi-blk-e} \\
\Phi_{11}(k, s) \Phi_{22}(k, s)\t &- \Phi_{12}(k, s) \Phi_{21}(k, s)\t = I_{n}. \label{eqn:phi-blk-f}
\end{align}
\end{subequations}
Moreover, $\Phi_{11}(k, s)$ and $\Phi_{22}(k, s)$ are invertible with 
\begin{equation} \label{eqn:phi-11-22}
\Phi_{11}(k, s) = \Phi_{22}(s, k)\t. 
\end{equation}
Furthermore, 
\begin{subequations} \label{eqn:phi-12-21}
\begin{align}
\Phi_{12}(k, s) &= - \Phi_{12}(s, k)\t, \label{eqn:phi-12-equiv} \\
\Phi_{21}(k, s) &= - \Phi_{21}(s, k)\t. \label{eqn:phi-21-equiv}
\end{align}
\end{subequations}
\end{lemma}

\begin{proof}
First, we show \eqref{eqn:phi-blk}. 
When $s = k$, we have $\Phi_{11}(s, s) = \Phi_{22}(s, s) = I_{n}$ and $\Phi_{12}(s, s) = \Phi_{21}(s, s) = 0$. 
Clearly, for $s = k$, \eqref{eqn:phi-blk} holds. 
Let 
\begin{equation*}
\bar{J} \triangleq 
\begin{bmatrix}
0 & I_{n} \\
-I_{n} & 0
\end{bmatrix}.
\end{equation*}
It can be checked that $M_{k}\t \bar{J} M_{k} = \bar{J} = M_{k} \bar{J} M_{k}\t$. 
When $s < k$, we have $\Phi_{M}(k, s) = M_{k-1} \Phi_{M}(k-1, s)$. 
It follows immediately that 
$\Phi_{M}(k, s)\t \bar{J} \Phi_{M}(k, s) = \Phi_{M}(k-1, s)\t M_{k-1}\t \bar{J} M_{k-1} \Phi_{M}(k-1, s) = \Phi_{M}(k-1, s)\t \bar{J} \Phi_{M}(k-1, s)$. 
Since, for all $s < k$, the above equation holds, we have by induction that $\Phi_{M}(k, s)\t \bar{J} \Phi_{M}(k, s) = \Phi_{M}(s, s)\t \bar{J} \Phi_{M}(s, s) = \bar{J}$. 
Similarly, 
$
\Phi_{M}(k, s) \bar{J} \Phi_{M}(k, s)\t 
= 
M_{k-1} \cdots M_{s} \Phi_{M}(s, s) \bar{J} \Phi_{M}(s, s)\t M_{s}\t \cdots M_{k-1}\t = \bar{J}
$. 
Thus, for $s < k$, \eqref{eqn:phi-blk} holds. 
Similarly, it can be shown that, for $s > k$, \eqref{eqn:phi-blk} holds.

Then, we show that $\Phi_{11}(k, s)$ is invertible. 
Let $\Pi_{k}$ be the solution to \eqref{eqn:pi-Q} with $\Pi_{0} = 0_{n \times n}$, and let $\mathcal{I}_{0} \subseteq \mathbb{Z}$ be the maximal integer time interval of existence of $\Pi_{k}$. 
Let $\bar{\Pi}_{k}$ be the solution to \eqref{eqn:pi-Q} when $Q_{k} \equiv 0_{n \times n}$ and $\bar{\Pi}_{0} = 0_{n \times n}$. 
Clearly, we have, for all $k \in \mathbb{Z}$, $\bar{\Pi}_{k} \equiv 0_{n \times n}$. 
Let $\hat{\Pi}_{k}$ be the solution to \eqref{eqn:pi-Q} when $B_{k} \equiv 0_{n \times n}$ and $\hat{\Pi}_{0} = 0_{n \times n}$. 
Since when $B_{k} \equiv 0_{n \times n}$, \eqref{eqn:pi-Q} is a linear matrix equation, 
it follows that, for all $k \in \mathbb{Z}$, the solution $\hat{\Pi}_{k}$ exists. 
Moreover, it is easy to see that, for $k < 0$, $\hat{\Pi}_{k} \succeq 0$, and, for $k > 0$, $\hat{\Pi}_{k} \preceq 0$. 
From the monotonicity of the matrix Riccati difference equation in Proposition~\ref{prp:pi-monotone}, we have
\begin{align*}
&0 \equiv \bar{\Pi}_{k} \preceq \Pi_{k} \preceq \hat{\Pi}_{k}, \quad k \in \mathcal{I}_{0}, ~ k \leq 0, \\
&\hat{\Pi}_{k} \preceq \Pi_{k} \preceq \bar{\Pi}_{k} \equiv 0, \quad k \in \mathcal{I}_{0}, ~ k \geq 0.
\end{align*}
Hence, $\Pi_{k}$ has no finite escape time and $\mathcal{I}_{0} = \mathbb{Z}$. 
From Proposition~\ref{prp:pi-Q-exist}, for all $k \in \mathbb{Z}$, $\Phi_{11}(k, 0)$ is invertible. 
A similar argument shows that, for all $k, s \in \mathbb{Z}$, $\Phi_{11}(k, s)$ is invertible.

Next, we show \eqref{eqn:phi-11-22}. 
From \eqref{eqn:phi-blk-f} and \eqref{eqn:phi-blk-c}, it follows that
\begin{align*}
\Phi_{11}(k, s)^{-1} 
&= \Phi_{22}(k, s)\t - \Phi_{12}(k, s)\t \Phi_{11}(k, s)^{- \mbox{\tiny\sf T}} \Phi_{21}(k, s)\t. 
\end{align*}
Since $\Phi_{11}(k, s)$ is invertible, the Schur complement of the block matrix $\Phi_{11}(k, s)$ in the matrix $\Phi_{M}(k, s)$ is given by 
$\Phi_{22}(k, s) - \Phi_{21}(k, s) \Phi_{11}(k, s)^{-1} \Phi_{12}(k, s) = \Phi_{22}(s, k)^{-1}$. 
Therefore, we have $\Phi_{11}(k, s) = \Phi_{22}(s, k)\t$.

Lastly, we show \eqref{eqn:phi-12-21}. 
It follows from equation \eqref{eqn:phi-blk-c} and the invertibility of $\Phi_{11}(k, s)$ that 
\begin{equation*}
\Phi_{12}(k, s) 
= \Phi_{22}(s, k)\t \Phi_{12}(k, s)\t \Phi_{11}(k, s)^{- \mbox{\tiny\sf T}}. 
\end{equation*}
Since $\Phi_{M}(k, s) \Phi_{M}(s, k) = I_{2n}$, it follows immediately that $\Phi_{11}(k, s) \Phi_{12}(s, k) + \Phi_{12}(k, s) \Phi_{22}(s, k) = 0$. 
Thus, 
\begin{equation*}
- \Phi_{12}(s, k)\t 
= \Phi_{22}(s, k)\t \Phi_{12}(k, s)\t \Phi_{11}(k, s)^{- \mbox{\tiny\sf T}} 
= \Phi_{12}(k, s). 
\end{equation*}
Similarly, it follows from equation \eqref{eqn:phi-blk-d} and the invertibility of $\Phi_{22}(k, s)$ that 
\begin{equation*}
\Phi_{21}(k, s)
= \Phi_{11}(s, k)\t \Phi_{21}(k, s)\t \Phi_{22}(k, s)^{- \mbox{\tiny\sf T}}. 
\end{equation*}
Since $\Phi_{21}(k, s) \Phi_{11}(s, k) + \Phi_{22}(k, s) \Phi_{21}(s, k) = 0$, we have 
\begin{equation*}
- \Phi_{21}(s, k)\t 
= \Phi_{11}(s, k)\t \Phi_{21}(k, s)\t \Phi_{22}(k, s)^{- \mbox{\tiny\sf T}} 
= \Phi_{21}(k, s). 
\end{equation*}
This completes the proof. \hfill
\end{proof}

The next result plays a critical role in the proof of Theorem~\ref{thm:pi-pos-exist-cond}.

\begin{lemma} \label{lem:phi-21-neg}
Suppose Assumption~\ref{asm:A-inv} holds. 
For all $k = 0, 1, \ldots, N$, 
\begin{equation} \label{eqn:phi-21-11-neg}
\Phi_{21}(k, 0) \Phi_{11}(k, 0)^{-1} = \Phi_{11}(k, 0)^{- \mbox{\tiny\sf T}} \Phi_{21}(k, 0)\t \preceq 0.
\end{equation}
For all $k = 1, 2, \ldots, N$, 
\begin{equation} \label{eqn:phi-mono}
- \Phi_{11}(k, 0)^{-1} \Phi_{12}(k, 0) 
= \Phi_{12}(0, k) \Phi_{22}(0, k)^{-1} 
= \sum_{i=0}^{k-1} X_{i},
\end{equation}
where 
\begin{align*}
X_{i} &\triangleq \Phi_{11}(i, 0)^{-1} A_{i}^{-1} B_{i} 
\Big(I_{p} + B_{i}\t A_{i}^{- \mbox{\tiny\sf T}} \big(Q_{i} - \Phi_{21}(i, 0) 
\\* 
&\hspace{10mm} \times \Phi_{11}(i, 0)^{-1}\big) A_{i}^{-1} B_{i}\Big)^{-1} 
B_{i}\t A_{i}^{- \mbox{\tiny\sf T}} \Phi_{11}(i, 0)^{- \mbox{\tiny\sf T}} \succeq 0.
\end{align*}
For all $k = 0, 1, \ldots, N-1$, 
\begin{equation} \label{eqn:phi-11-img}
\sum_{i=0}^{k} \range \hspace{-0.5mm} \Big( \hspace{-0.5mm} \Phi_{11}(i, 0)^{-1} A_{i}^{-1} B_{i} \hspace{-0.5mm} \Big) \hspace{-0.6mm} = \hspace{-0.8mm} \sum_{i=0}^{k} \range \hspace{-0.5mm} \Big( \hspace{-0.5mm} \Phi_{A}(0, i) A_{i}^{-1} B_{i} \hspace{-0.5mm} \Big).
\end{equation}
\end{lemma}

\begin{proof}
From \eqref{eqn:phi-blk-b}, we have $\Phi_{21}(k, s)\t \Phi_{11}(k, s) = \Phi_{11}(k, s)\t \Phi_{21}(k, s)$. 
From Lemma~\ref{lem:phi-eqn}, $\Phi_{11}(k, s)$ is invertible. 
It follows that 
$\Phi_{21}(k, 0) \Phi_{11}(k, 0)^{-1} = \Phi_{11}(k, 0)^{- \mbox{\tiny\sf T}} \Phi_{21}(k, 0)\t$.

Then, we show that $\Phi_{21}(k, 0) \Phi_{11}(k, 0)^{-1} \preceq 0$ by induction. 
When $k = 0$, $\Phi_{21}(0, 0) \Phi_{11}(0, 0)^{-1} = 0$. 
For $k = i$, assume $\Phi_{21}(i, 0) \Phi_{11}(i, 0)^{-1} \preceq 0$. 
It follows from $\Phi_{M}(i+1, 0) = M_{i} \Phi_{M}(i, 0)$ that
\begin{subequations} \label{eqn:phi-11-21}
\begin{align}
&\Phi_{11}(i+1, 0) = 
\label{eqn:phi-11} \\* 
&A_{i} \Phi_{11}(i, 0) \hspace{-0.5mm} + \hspace{-0.5mm} B_{i} B_{i}\t A_{i}^{- \mbox{\tiny\sf T}} \big(Q_{i} - \Phi_{21}(i, 0) \Phi_{11}(i, 0)^{-1}\big) \Phi_{11}(i, 0), \nonumber \\
&\Phi_{21}(i+1, 0) = - A_{i}^{- \mbox{\tiny\sf T}} \big(Q_{i} - \Phi_{21}(i, 0) \Phi_{11}(i, 0)^{-1}\big) \Phi_{11}(i, 0). \label{eqn:phi-21}
\end{align}
\end{subequations}
Thus, from \eqref{eqn:phi-11-21}, we have
\begin{align*}
&\Phi_{11}(i+1, 0)\t \Phi_{21}(i+1, 0) = 
- \Phi_{11}(i, 0)\t \big(Q_{i} - \Phi_{21}(i, 0) 
\\* 
&\hspace{-1mm} 
\times \hspace{-0.7mm} \Phi_{11}(i, 0)^{-1} \hspace{-0.4mm} \big) \Phi_{11}(i, 0) 
\hspace{-0.7mm} - \hspace{-0.7mm} \Phi_{11}(i, 0)\t \big( \hspace{-0.4mm} Q_{i} \hspace{-0.7mm} - \hspace{-0.7mm} \Phi_{21}(i, 0) \Phi_{11}(i, 0)^{-1} \hspace{-0.4mm} \big) 
\\*
&\hspace{-1mm} 
\times 
A_{i}^{-1} B_{i} B_{i}\t A_{i}^{- \mbox{\tiny\sf T}} \big(Q_{i} - \Phi_{21}(i, 0) \Phi_{11}(i, 0)^{-1}\big) \Phi_{11}(i, 0).
\end{align*}
Since $\Phi_{21}(i, 0) \Phi_{11}(i, 0)^{-1} \preceq 0$ and $Q_{i} \succeq 0$, it follows that $\Phi_{11}(i+1, 0)\t \Phi_{21}(i+1, 0) \preceq 0$. 
Therefore, it follows that $\Phi_{21}(i+1, 0) \Phi_{11}(i+1, 0)^{-1} \preceq 0$.

Next, we show \eqref{eqn:phi-mono}. 
From Lemma \ref{lem:phi-eqn}, $\Phi_{11}(k, 0)$ and $\Phi_{22}(0, k)$ are invertible. 
Multiplying 
$\Phi_{11}(k, 0) \Phi_{12}(0, k) + \Phi_{12}(k, 0) \Phi_{22}(0, k) = 0$ 
by $\Phi_{11}(k, 0)^{-1}$ on the left and by $\Phi_{22}(0, k)^{-1}$ on the right yields 
$- \Phi_{11}(k, 0)^{-1} \Phi_{12}(k, 0) = \Phi_{12}(0, k) \Phi_{22}(0, k)^{-1}$.

It follows from \eqref{eqn:phi-21-11-neg} that $X_{i} \succeq 0$. 
For all $k = 0, 1, \ldots, N-1$, we can compute that, 
\begin{equation*}
- \Phi_{12}(k+1, 0) \Phi_{11}(k, 0)\t 
= 
- \Phi_{11}(k+1, 0) \Phi_{12}(k, 0)\t - M_{k, 2},
\end{equation*}
where $M_{k, 1}$ and $M_{k, 2}$ are given in \eqref{def:M}. 
In view of Lemma~\ref{lem:phi-eqn}, we have
\begin{align*}
& - \Phi_{11}(k+1, 0)^{-1} \Phi_{12}(k+1, 0) = - \Phi_{11}(k, 0)^{-1} \Phi_{12}(k, 0) 
\\* 
&\hspace{22mm} 
+ \Phi_{11}(k+1, 0)^{-1} B_{k} B_{k}\t A_{k}^{- \mbox{\tiny\sf T}} \Phi_{11}(k, 0)^{- \mbox{\tiny\sf T}}.
\end{align*}

It is easy to check that
\begin{align*}
&\Big(A_{k} + B_{k} B_{k}\t A_{k}^{- \mbox{\tiny\sf T}} \big(Q_{k} - \Phi_{21}(k, 0) \Phi_{11}(k, 0)^{-1}\big)\Big) A_{k}^{-1} B_{k} \\
&= B_{k} \Big(I_{p} + B_{k}\t A_{k}^{- \mbox{\tiny\sf T}} \big(Q_{k} - \Phi_{21}(k, 0) \Phi_{11}(k, 0)^{-1}\big) A_{k}^{-1} B_{k}\Big).
\end{align*}
It follows from \eqref{eqn:phi-11} and the above equation that 
\begin{equation*}
\Phi_{11}(k+1, 0)^{-1} B_{k} B_{k}\t A_{k}^{- \mbox{\tiny\sf T}} \Phi_{11}(k, 0)^{- \mbox{\tiny\sf T}} = X_{k}.
\end{equation*}

\vspace{-2mm}
Lastly, we show by induction the following stronger statements, for all $k = 0, 1, \ldots, N-1$, 
\begin{align}
& \sum_{i=0}^{k} \range \hspace{-0.5mm} \Big( \hspace{-0.5mm} \Phi_{11}(i, 0)^{-1} A_{i}^{-1} B_{i} \hspace{-0.5mm} \Big) \hspace{-0.6mm} = \hspace{-0.8mm} \sum_{i=0}^{k} \range \hspace{-0.5mm} \Big( \hspace{-0.5mm} \Phi_{A}(0, i) A_{i}^{-1} B_{i} \hspace{-0.5mm} \Big), \nonumber \\
& \Phi_{11}(k, 0)^{-1} \hspace{-0.5mm} = \hspace{-0.5mm} \Phi_{A}(0, k) \hspace{-0.5mm} + \hspace{-0.5mm} \sum_{i=0}^{k-1} \Phi_{11}(i, 0)^{-1} A_{i}^{-1} B_{i} Z_{i, k}, \label{eqn:phi-11-A} 
\end{align}
for some $Z_{i, k} \in \mathbb{R}^{p \times n}$,
with the convention that, for $k = 0$, \eqref{eqn:phi-11-A} reduces to $\Phi_{11}(0, 0)^{-1} = \Phi_{A}(0, 0) = I_{n}$.

When $k = 0$, we know that \eqref{eqn:phi-11-A} holds and we can check that $\Phi_{11}(0, 0)^{-1} A_{0}^{-1} B_{0} = \Phi_{A}(0, 0) A_{0}^{-1} B_{0}$, so their ranges are equal, and thus \eqref{eqn:phi-11-img} holds.

Assume now that, for all $k = 0, 1, \ldots, j$, where $j \leq N-2$, \eqref{eqn:phi-11-img} and \eqref{eqn:phi-11-A} hold. 
Since $\Phi_{M}(j+1, 0) = M_{j} \Phi_{M}(j, 0)$, we have $\Phi_{M}(j, 0) = M_{j}^{-1} \Phi_{M}(j+1, 0)$, where $M_{j}^{-1}$ is given by \eqref{eqn:M-inv}. 
Hence, $\Phi_{11}(j, 0) = A_{j}^{-1} \Phi_{11}(j+1, 0) + A_{j}^{-1} B_{j} B_{j}\t \Phi_{21}(j+1, 0)$. 
It follows that 
\begin{align*}
&\Phi_{11}(j+1, 0)^{-1} 
= \Big(A_{j} \Phi_{11}(j, 0) - B_{j} B_{j}\t \Phi_{21}(j+1, 0)\Big)^{-1} \\
&\hspace{-0.5mm} = \Phi_{11}(j, 0)^{-1} A_{j}^{-1} + \Phi_{11}(j, 0)^{-1} A_{j}^{-1} B_{j} Y_{j}, ~ \text{for some}~ Y_{j}, \\
&\hspace{-0.5mm} = \hspace{-0.5mm} \Phi_{A} \hspace{-0.3mm} ( \hspace{-0.3mm} 0, j \hspace{-0.5mm} + \hspace{-0.5mm} 1 \hspace{-0.3mm} ) \hspace{-0.6mm} + \hspace{-1mm} \sum_{i=0}^{j} \hspace{-0.6mm} \Phi_{11} \hspace{-0.3mm} ( \hspace{-0.3mm} i, 0 \hspace{-0.3mm} )^{-1} \hspace{-0.8mm} A_{i}^{-1} \hspace{-0.8mm} B_{i} Z_{i, j+1}, \text{for some}~ Z_{i, j+1}.
\end{align*}
Thus, for $k = j+1$, \eqref{eqn:phi-11-A} holds. 
Next, we show \eqref{eqn:phi-11-img}. 
Since
\begin{multline*} 
\Phi_{11}(j+1, 0)^{-1} A_{j+1}^{-1} B_{j+1} 
= \Phi_{A}(0, j+1) A_{j+1}^{-1} B_{j+1} 
\\* 
+ \sum_{i=0}^{j} \Phi_{11}(i, 0)^{-1} A_{i}^{-1} B_{i} Z_{i, j+1} A_{j+1}^{-1} B_{j+1},
\end{multline*}
it follows from the above equation and the induction assumption that 
$
\range \Big(\Phi_{11}(j+1, 0)^{-1} A_{j+1}^{-1} B_{j+1}\Big) \subseteq 
\sum_{i=0}^{j+1} \range \Big(\Phi_{A}(0, i) A_{i}^{-1} B_{i}\Big)
$, 
and 
$
\range \Big(\Phi_{A}(0, j+1) A_{j+1}^{-1} B_{j+1}\Big) \subseteq 
\sum_{i=0}^{j+1} \range \Big(\Phi_{11}(i, 0)^{-1} A_{i}^{-1} B_{i}\Big)
$. 
Therefore, for $k = j+1$, \eqref{eqn:phi-11-img} holds. \hfill
\end{proof}

\vspace{-3mm}
We are now ready to show the positive definiteness of $- \Phi_{11}(N, 0)^{-1} \Phi_{12}(N, 0)$, which will be utilized as an upper bound on $\Pi_{0}$ for the subsequent results.

\begin{corollary} \label{cor:phi-12-inv}
Under Assumption~\ref{asm:A-inv} and Assumption~\ref{asm:G-inv}, 
$- \Phi_{11}(N, 0)^{-1} \Phi_{12}(N, 0) \succ 0$. 
It follows that $\Phi_{12}(N, 0)$ is invertible. 
\end{corollary}

\vspace{-5mm}
\begin{proof}
From \eqref{eqn:phi-mono}, $- \Phi_{11}(N, 0)^{-1} \Phi_{12}(N, 0) \succeq 0$. 
Assume there exists $\zeta \in \mathbb{R}^{n}$ with $\zeta \neq 0$ such that $\zeta\t \big(- \Phi_{11}(N, 0)^{-1} \Phi_{12}(N, 0)\big) \zeta = 0$. 
It follows that, for all $k = 0, 1, \ldots, N-1$, $\zeta\t X_{k} \zeta = 0$. 
The expression of $X_{k}$ implies that, for all $k = 0, 1, \ldots, N-1$, $\zeta\t \Phi_{11}(k, 0)^{-1} A_{k}^{-1} B_{k} = 0$. 
In view of \eqref{eqn:phi-11-img}, for all $k = 0, 1, \ldots, N-1$, we have $\zeta\t A_{0}^{-1} A_{1}^{-1} \ldots A_{k}^{-1} B_{k} = 0$. 
Thus, $\zeta\t G(0, N) \zeta = - \zeta\t \Phi_{A}(0, N) G(N, 0) \Phi_{A}(0, N)\t \zeta = 0$, which contradicts either 
Assumption~\ref{asm:A-inv} or Assumption~\ref{asm:G-inv}. \hfill
\end{proof}


\section*{Appendix~D}

In this appendix, we prove Theorem~\ref{thm:pi-pos-exist-cond}, which provides necessary and sufficient conditions for the solution of the Riccati difference equation \eqref{eqn:pi-Q} to satisfy Property~\ref{pty:pi-pos}. 
The first result is an alternative statement of Property~\ref{pty:pi-pos}.

\begin{lemma} \label{lem:pi-pos}
Under Assumption~\ref{asm:A-inv}, suppose the condition in Proposition~\ref{prp:pi-Q-exist} holds so that \eqref{eqn:pi-Q} admits a solution $\Pi_{k}$ from time $0$ to $N$. 
Then, for all $k = 0, 1, \ldots, N-1$, we have that 
\begin{equation} \label{eqn:pi-pos}
I_{p} + B_{k}\t \Pi_{k+1} B_{k} = \big(I_{p} - B_{k}\t A_{k}^{- \mbox{\tiny\sf T}} \big(\Pi_{k} - Q_{k}\big) A_{k}^{-1} B_{k}\big)^{-1}.
\end{equation}
Thus, $I_{p} + B_{k}\t \Pi_{k+1} B_{k} \succ 0$ if and only if $I_{p} - B_{k}\t A_{k}^{- \mbox{\tiny\sf T}} \big(\Pi_{k} - Q_{k}\big) A_{k}^{-1} B_{k} \succ 0$. 
Moreover, equation \eqref{sol:pi-Q-0} can be written, equivalently, as
\begin{multline} \label{eqn:pi-k-alt}
\Pi_{k} = \Phi_{21}(k, 0) \Phi_{11}(k, 0)^{-1} + \Phi_{11}(k, 0)^{- \mbox{\tiny\sf T}} \Pi_{0} 
\\* 
\times \Big(I_{n} + \Phi_{11}(k, 0)^{-1} \Phi_{12}(k, 0) \Pi_{0}\Big)^{-1} \Phi_{11}(k, 0)^{-1}. 
\end{multline}
\end{lemma}

\begin{proof}
In view of \eqref{eqn:phi-M-blk} and \eqref{sol:pi-Q-s}, we can compute that
\begin{equation*}
\Pi_{k+1} = 
A_{k}^{- \mbox{\tiny\sf T}} \big(\Pi_{k} - Q_{k}\big) \big(A_{k} - B_{k} B_{k}\t A_{k}^{- \mbox{\tiny\sf T}} \big(\Pi_{k} - Q_{k}\big) \big)^{-1}.
\end{equation*}
It follows that 
\begin{equation*}
I_{p} + B_{k}\t \Pi_{k+1} B_{k} = 
\Big(I_{p} - B_{k}\t A_{k}^{- \mbox{\tiny\sf T}} \big(\Pi_{k} - Q_{k}\big) A_{k}^{-1} B_{k}\Big)^{-1}.
\end{equation*}

In view of \eqref{eqn:phi-blk-b}, \eqref{eqn:phi-blk-e}, and the invertibility of $\Phi_{11}(k, 0)$, we have 
$
\Phi_{21}(k, 0) + \Phi_{22}(k, 0) \Pi_{0} = 
\Phi_{11}(k, 0)^{- \mbox{\tiny\sf T}} \Big(\Pi_{0} + \Phi_{21}(k, 0)\t \big(\Phi_{11}(k, 0) + \Phi_{12}(k, 0) \Pi_{0}\big)\Big)
$. 
From \eqref{sol:pi-Q-0}, the above equation, \eqref{eqn:phi-blk-b}, and the invertibility of $\Phi_{11}(k, 0)$, it follows that \eqref{eqn:pi-k-alt} holds. \hfill
\end{proof}

The second set of necessary and sufficient conditions for Property~\ref{pty:pi-pos} is now stated as follows.

\begin{proposition} \label{prp:U-pos}
Under Assumption~\ref{asm:A-inv}, the following statements are equivalent.
\begin{enumerate}[label=\roman*)]
\item Property~\ref{pty:pi-pos} holds. 
\item For all $k = 1, 2, \ldots, N-1$, $U_{0} \succ 0$ and $U_{k}/U_{k-1} \succ 0$. 
\item For all $k = 0, 1, \ldots, N-1$, $U_{k} \succ 0$. 
\item $U_{N-1} \succ 0$.
\item For all $k = 0, 1, \ldots, N$, the eigenvalues of $I_{n} + \Phi_{11}(k, 0)^{-1} \Phi_{12}(k, 0) \Pi_{0}$ are positive. 
\item The eigenvalues of $I_{n} + \Phi_{11}(N, 0)^{-1} \Phi_{12}(N, 0) \Pi_{0}$ are positive. 
\end{enumerate}
\end{proposition}

\begin{proof}
First, we claim that 
$U_{0} = I_{p} - B_{0}\t A_{0}^{- \mbox{\tiny\sf T}} \big(\Pi_{0} - Q_{0}\big) A_{0}^{-1} B_{0}$, 
and, for all $k = 1, 2, \ldots, N-1$, the Schur complement of the block $U_{k-1}$ in the matrix $U_{k}$ is 
$U_{k}/U_{k-1} = I_{p} - B_{k}\t A_{k}^{- \mbox{\tiny\sf T}} \big(\Pi_{k} - Q_{k}\big) A_{k}^{-1} B_{k}$. 
Then, the equivalence of $i)$ and $ii)$ follows from \eqref{eqn:pi-pos} and the above claim. 
The equivalence of $ii)$, $iii)$, and $iv)$ follows from Sylvester's criterion~\cite{horn2012matrix}.

Next, for $k = 0$, $v)$ holds. 
In light of \eqref{eqn:phi-mono}, for all $k = 1, 2, \ldots, N$, we have 
$- \Phi_{11}(k, 0)^{-1} \Phi_{12}(k, 0) 
= \sum_{i=0}^{k-1} X_{i} 
= \sum_{i=0}^{k-1} L_{i} T_{i}^{-1} L_{i}\t$. 
It follows that, 
\begin{multline*}
- \Phi_{11}(k+1, 0)^{-1} \Phi_{12}(k+1, 0) 
= 
\\* 
\begin{bmatrix}
L_{k} & \cdots & L_{0}
\end{bmatrix}
\begin{bmatrix}
T_{k}^{-1} &    0   &   0   \\
      0      & \ddots &   0   \\
      0      &    0   & T_{0}^{-1} 
\end{bmatrix}
\begin{bmatrix}
L_{k}\t \\
\vdots \\ 
L_{0}\t
\end{bmatrix}. 
\end{multline*}
Hence, for all $k = 0, 1, \ldots, N-1$, all eigenvalues of $I_{n} + \Phi_{11}(k+1, 0)^{-1} \Phi_{12}(k+1, 0) \Pi_{0}$ are positive if and only if $U_{k} \succ 0$.

Lastly, we show the claim. 
For notational simplicity, below we will use the notation $\Phi_{ij}^{k,0} \triangleq \Phi_{ij}(k, 0)$, where $i, j \in \{1, 2\}$. 
We first check that 
$U_{0} 
= T_{0} - L_{0}\t \Pi_{0} L_{0} 
= I_{p} - B_{0}\t A_{0}^{- \mbox{\tiny\sf T}} \big(\Pi_{0} - Q_{0}\big) A_{0}^{-1} B_{0}$. 
Then, we compute that 
\begin{align*}
&U_{k-1}^{-1} = 
\begin{bmatrix}
T_{k-1}^{-1} &    0   &   0   \\
      0      & \ddots &   0   \\
      0      &    0   & T_{0}^{-1} 
\end{bmatrix}
\hspace{-1mm} + \hspace{-1mm}
\begin{bmatrix}
T_{k-1}^{-1} &    0   &   0   \\
      0      & \ddots &   0   \\
      0      &    0   & T_{0}^{-1} 
\end{bmatrix}
\hspace{-2mm}
\begin{bmatrix}
L_{k-1}\t \\
\vdots \\ 
L_{0}\t
\end{bmatrix}
\hspace{-1mm}
\Pi_{0} 
\\* 
&\hspace{-0.5mm} \times \hspace{-1mm}
\Big( \hspace{-0.8mm} I_{n} \hspace{-1mm} + \hspace{-1mm} \big(\Phi_{11}^{k,0}\big)^{-1} \hspace{-0.6mm} \Phi_{12}^{k,0} \Pi_{0} \hspace{-0.8mm} \Big)^{-1} 
\hspace{-0.6mm}
\Big[
\hspace{-0.2mm}
L_{k-1} ~ \cdots ~ L_{0}
\hspace{-0.2mm}
\Big]
\hspace{-1.6mm}
\begin{bmatrix}
\hspace{-0.2mm}
T_{k-1}^{-1} &    0   &   0   \\
      0      & \ddots &   0   \\
      0      &    0   & T_{0}^{-1} 
\hspace{-0.5mm}
\end{bmatrix} \hspace{-1.5mm} .
\end{align*}
It follows that the desired Schur complement is 
\begin{equation*}
U_{k}/U_{k-1} = 
I_{p} - B_{k}\t A_{k}^{- \mbox{\tiny\sf T}} \big(\Pi_{k} - Q_{k}\big) A_{k}^{-1} B_{k},
\end{equation*}
where we have used \eqref{eqn:pi-k-alt}. \hfill
\end{proof}

With the addition of Assumption~\ref{asm:G-inv}, we can derive the following necessary and sufficient condition for Property~\ref{pty:pi-pos}.

\begin{corollary} \label{cor:pi-0}
Under Assumptions~\ref{asm:A-inv} and \ref{asm:G-inv}, Property~\ref{pty:pi-pos} holds if and only if 
$
\Pi_{0} \prec - \Phi_{12}(N, 0)^{-1} \Phi_{11}(N, 0)
$. 
\end{corollary}

\begin{proof}
From Corollary~\ref{cor:phi-12-inv}, $- \Phi_{11}(N, 0)^{-1} \Phi_{12}(N, 0) \succ 0$. 
Statement $vi)$ in Proposition~\ref{prp:U-pos} is equivalent to 
$I_{n} - \big(- \Phi_{11}(N, 0)^{-1} \Phi_{12}(N, 0)\big)^{\frac{1}{2}} \Pi_{0} \big(- \Phi_{11}(N, 0)^{-1} \Phi_{12}(N, 0)\big)^{\frac{1}{2}} \succ 0$, 
which is equivalent to 
$\big(- \Phi_{11}(N, 0)^{-1} \Phi_{12}(N, 0)\big)^{-1} - \Pi_{0} \succ 0$. 
In view of Proposition~\ref{prp:U-pos}, Property~\ref{pty:pi-pos} holds if and only if $\Pi_{0} \prec - \Phi_{12}(N, 0)^{-1} \Phi_{11}(N, 0)$. \hfill
\end{proof}


\section*{Appendix~E}

In this appendix, we provide two auxiliary results for proving Proposition~\ref{prp:map-inv-Q} and Theorem~\ref{thm:mean-covar}, respectively. 
The following matrix inequality is used in the proof of Proposition~\ref{prp:map-inv-Q}.

\begin{lemma} \label{lem:phi-11-12-inv-mono}
Under Assumptions~\ref{asm:A-inv} and \ref{asm:G-inv}, suppose $\Pi_{0} \prec - \Phi_{12}(N, 0)^{-1} \Phi_{11}(N, 0)$. 
Then, for all $N \geq k \geq i \geq 0$, 
\begin{equation} \label{ineq:phi-11-12-inv-mono}
- W_{k,0} \succeq - W_{i,0} \succeq 0.
\end{equation}
\end{lemma}

\begin{proof}
For notational simplicity, we define the matrix 
$\Theta_{k} \triangleq - \Phi_{11}(k, 0)^{-1} \Phi_{12}(k, 0) \succeq 0$. 
For $i = 0$, we can check that $\big(I_{n} - \Theta_{0} \Pi_{0}\big)^{-1} \Theta_{0} = 0$. 
Next, assume $i \geq 1$. 
It follows from \eqref{eqn:phi-mono} that $\Theta_{k} \succeq \Theta_{i} \succeq 0$. 
This implies that $\ker \Theta_{k} \subseteq \ker \Theta_{i}$. 
If $\Theta_{k} \succeq \Theta_{i} \succ 0$, we have $\big(\Theta_{k}^{-1} - \Pi_{0}\big)^{-1} \succeq \big(\Theta_{i} - \Pi_{0}\big)^{-1} \succ 0$. 
Hence, \eqref{ineq:phi-11-12-inv-mono} holds.

Without loss of generality, assume that $\Theta_{k}$ and $\Theta_{i}$ are singular and $\ker \Theta_{k} \subsetneq \ker \Theta_{i}$. 
The case when $\Theta_{k} \succ 0$ or when $\ker \Theta_{k} = \ker \Theta_{i}$ can be proved in a similar way. 
Since $\ker \Theta_{k} \subseteq \ker \Theta_{i}$, there exists an $n \times n$ orthogonal matrix $U$ such that\footnote{The entries denoted by $\times$ do not affect the logic of the proof and thus are hidden.}
\begin{align*}
\Theta_{k} &= 
U 
\begin{bmatrix}
\bar{\Theta}_{k} & 0 \\
0 & 0 
\end{bmatrix}
U\t,
~~~~~ 
\Theta_{i} = 
U 
\begin{bmatrix}
\bar{\Theta}_{i} & 0 \\
0 & 0 
\end{bmatrix}
U\t,
\\ 
\Theta_{N}^{-1} &= 
U 
\begin{bmatrix}
\bar{\Theta}_{N}^{-1} & \times \\
\times & \times 
\end{bmatrix}
U\t,
~~~ 
\Pi_{0} = 
U 
\begin{bmatrix}
\bar{\Pi}_{0} & \times \\
\times & \times 
\end{bmatrix}
U\t,
\end{align*}
where $\bar{\Theta}_{k} \succ 0$, $\bar{\Theta}_{i} \succeq 0$, $\bar{\Theta}_{N} \succ 0$, and $\bar{\Pi}_{0}$ are $r \times r$ matrices. 
Since $\Theta_{N}^{-1} \succ \Pi_{0}$, we have $\bar{\Theta}_{N}^{-1} \succ \bar{\Pi}_{0}$. 
Since $\Theta_{N} \succeq \Theta_{k} \succeq 0$, Lemma~\ref{lem:monotone} implies that $\bar{\Theta}_{N} \succeq \bar{\Theta}_{k} \succ 0$. 
It follows that $\bar{\Theta}_{k}^{-1} \succeq \bar{\Theta}_{N}^{-1} \succ \bar{\Pi}_{0}$. 
In view of \eqref{ineq:phi-11-12-inv-mono}, it suffices to show that 
\begin{equation} \label{ineq:phi-11-12-inv-mono-r}
\Big(\bar{\Theta}_{k}^{-1} - \bar{\Pi}_{0}\Big)^{-1} 
\succeq 
\Big(I_{r} - \bar{\Theta}_{i} \bar{\Pi}_{0}\Big)^{-1} \bar{\Theta}_{i} 
\succeq 0. 
\end{equation}

Since $\ker \Theta_{k} \neq \ker \Theta_{i}$, there exists an $r \times r$ orthogonal matrix $V$ such that
\begin{align*}
\bar{\Theta}_{k}^{-1} &= 
V 
\begin{bmatrix}
\hat{\Theta}_{k}^{-1} & \times \\
\times & \times 
\end{bmatrix}
V\t, 
\quad
\bar{\Theta}_{i} = 
V 
\begin{bmatrix}
\hat{\Theta}_{i} & 0 \\
0 & 0 
\end{bmatrix}
V\t, 
\\ 
\bar{\Pi}_{0} &= 
V
\begin{bmatrix}
\hat{\Pi}_{0} & \times \\
\times & \times 
\end{bmatrix}
V\t,
\end{align*}
where $\hat{\Theta}_{k} \succ 0$, $\hat{\Theta}_{i} \succ 0$, and $\hat{\Pi}_{0}$ are $\rho \times \rho$ matrices. 
Since $\bar{\Theta}_{k} \succeq \bar{\Theta}_{i} \succeq 0$, it follows from Lemma~\ref{lem:monotone} that $\hat{\Theta}_{k} \succeq \hat{\Theta}_{i} \succ 0$. 
It follows that $\hat{\Theta}_{i}^{-1} \succeq \hat{\Theta}_{k}^{-1} \succ \hat{\Pi}_{0}$. 
To this end, in view of \eqref{ineq:phi-11-12-inv-mono-r}, it suffices to show that 
\begin{equation*}
\begin{bmatrix}
\hat{\Theta}_{k}^{-1} - \hat{\Pi}_{0} & \times \\
\times & \times 
\end{bmatrix}^{-1}
\succeq 
\begin{bmatrix}
\Big(\hat{\Theta}_{i}^{-1} - \hat{\Pi}_{0}\Big)^{-1} & 0 \\
0 & 0 
\end{bmatrix} 
\succeq 0. 
\end{equation*}
Since $\bar{\Theta}_{k}^{-1} - \bar{\Pi}_{0} \succ 0$, in light of Lemma~\ref{lem:monotone}, it suffices to show that 
$
(\hat{\Theta}_{k}^{-1} - \hat{\Pi}_{0})^{-1} \succeq (\hat{\Theta}_{i}^{-1} - \hat{\Pi}_{0})^{-1} \succ 0
$. 
Hence, it is sufficient to show that 
$\hat{\Theta}_{i}^{-1} - \hat{\Pi}_{0} \succeq 
\hat{\Theta}_{k}^{-1} - \hat{\Pi}_{0} \succ 0$, 
which is a direct result of the fact that $\hat{\Theta}_{i}^{-1} \succeq \hat{\Theta}_{k}^{-1} \succ \hat{\Pi}_{0}$. \hfill
\end{proof}


The fact below is used to prove Theorem~\ref{thm:mean-covar}.

\begin{lemma} \label{lem:gramian-bar}
If, for all $k \in \mathbb{Z}$, $A_{k}$ is invertible, then, for all $k, s \in \mathbb{Z}$, 
\begin{equation} \label{eqn:gramian-bar}
\bar{G}(k, s) = - \Phi_{12}(k, s) \Phi_{\bar{A}}(k, s)\t. 
\end{equation}
\end{lemma}

\begin{proof}
The proof is by induction on $k$. 
For $k = s$, it is not difficult to see that \eqref{eqn:gramian-bar} holds.

Now assume that for $k = j$ \eqref{eqn:gramian-bar} holds. 
It can be verified that $\bar{G}(j+1, s) = \bar{B}_{j} \bar{B}_{j}\t + \bar{A}_{j} \bar{G}(j, s) \bar{A}_{j}\t$, regardless of $s \leq j$ or $s > j$. 
Thus, it suffices to show that 
\begin{equation*}
\bar{B}_{j} \bar{B}_{j}\t - \bar{A}_{j} \Phi_{12}(j, s) \Phi_{\bar{A}}(j, s)\t \bar{A}_{j}\t 
\hspace{-0.5mm} = \hspace{-0.5mm} 
- \Phi_{12}(j+1, s) \Phi_{\bar{A}}(j+1, s)\t \hspace{-0.1mm} . 
\end{equation*}
Since $\Phi_{M}(j, s) = M_{j}^{-1} \Phi_{M}(j+1, s)$, it follows from \eqref{eqn:M-inv}, by replacing $B_{k} B_{k}\t$ with $B_{k} R_{k}^{-1} B_{k}\t$, that $\Phi_{12}(j, s) = A_{j}^{-1} \Phi_{12}(j+1, s) + A_{j}^{-1} B_{j} R_{j}^{-1} B_{j}\t \Phi_{22}(j+1, s)$. 
Since $\Phi_{\bar{A}}(j, s)\t \bar{A}_{j}\t = \Phi_{\bar{A}}(j+1, s)\t$, in light of \eqref{def:A-bar} and \eqref{def:B-bar}, it suffices to show that 
$
B_{j} \big(R_{j} + B_{j}\t \Pi_{j+1} B_{j} \big)^{-1} B_{j}\t \Phi_{\bar{A}}(s, j+1)\t
= 
\big(I_{n} + B_{j} R_{j}^{-1} B_{j}\t \Pi_{j+1}\big)^{-1} \big(\Phi_{12}(j+1, s) +  B_{j} R_{j}^{-1} B_{j}\t \Phi_{22}(j+1, s)\big) - \Phi_{12}(j+1, s)
$. 
From Proposition~\ref{prp:phi-AB-pi-Q}, we have $\Phi_{\bar{A}}(s, j+1)\t = \Phi_{11}(s, j+1)\t + \Pi_{j+1} \Phi_{12}(s, j+1)\t$. 
From the Woodbury formula~\cite{horn2012matrix}, it follows that $\big(I_{n} + B_{j} R_{j}^{-1} B_{j}\t \Pi_{j+1}\big)^{-1} = I_{n} - B_{j} \big(R_{j} + B_{j}\t \Pi_{j+1} B_{j} \big)^{-1} B_{j}\t \Pi_{j+1}$. 
In view of \eqref{eqn:phi-12-equiv} and \eqref{eqn:phi-11-22}, it suffices to show that 
$
B_{j} \big(R_{j} + B_{j}\t \Pi_{j+1} B_{j} \big)^{-1} B_{j}\t \Phi_{11}(s, j+1)\t 
= 
\big(I_{n} + B_{j} R_{j}^{-1} B_{j}\t \Pi_{j+1}\big)^{-1} B_{j} R_{j}^{-1} B_{j}\t \Phi_{11}(s, j+1)\t
$. 
Since 
$
B_{j} \big(R_{j} + B_{j}\t \Pi_{j+1} B_{j} \big)^{-1} B_{j}\t - \big(I_{n} + B_{j} R_{j}^{-1} B_{j}\t \Pi_{j+1}\big)^{-1} B_{j} R_{j}^{-1} B_{j}\t 
= B_{j} R_{j}^{-1} B_{j}\t - B_{j} R_{j}^{-1} B_{j}\t 
= 0
$, 
the previous equation holds. 
Thus, for $k = j+1$, \eqref{eqn:gramian-bar} holds. \hfill
\end{proof}

\end{document}